\documentclass[pra,aps,twocolumn, nobalancelastpage, superscriptaddress]{revtex4-1}
\usepackage[utf8]{inputenc}
\usepackage[utf8]{inputenc}
\usepackage{bbold}
\usepackage{geometry}
\usepackage{float}
\usepackage{graphicx}
\graphicspath{ {images/} }
\usepackage{braket}
\usepackage[normalem]{ulem}
\usepackage[T1]{fontenc}
\usepackage{gensymb}
\usepackage{amsmath,amsfonts,amssymb}
\usepackage{amsmath}
\usepackage{amsthm}
\usepackage[]{natbib}
\usepackage{comment}

\newcommand{\be}{\begin{equation}}
\newcommand{\ee}{\end{equation}}
\newcommand{\ba}{\begin{aligned}}
\newcommand{\ea}{\end{aligned}}

\newcommand{\bc}{\begin{center}}
\newcommand{\ec}{\end{center}}
\newcommand{\beq}{\begin{equation}}
\newcommand{\eeq}{\end{equation}}
\newcommand{\beqq}{\begin{equation*}}
\newcommand{\eeqq}{\end{equation*}}
\newcommand{\beqa}{\begin{align}}
\newcommand{\eeqa}{\end{align}}
\newcommand{\barr}{\begin{array}}
\newcommand{\earr}{\end{array}}
\newcommand{\bi}{\begin{itemize}}
\newcommand{\ei}{\end{itemize}}

\newtheorem{lem}{Lemma}

\usepackage[utf8]{inputenc}
\usepackage[T1]{fontenc} %fontenc makes the pdf ugly
\usepackage{lmodern} % Nicer font fith T1 fontenc
\usepackage{microtype} % Even nicer typography
\usepackage{graphicx}
\usepackage{dcolumn}
\usepackage{enumerate,amsthm,amssymb,color}
\usepackage{amsfonts}
\usepackage{amsmath}
\usepackage{mathtools}
\usepackage{dsfont}

\usepackage{subcaption}
\usepackage{ragged2e}
\DeclareCaptionJustification{justified}{\justifying}
\usepackage{bbm}
\usepackage{bm}
\usepackage{synttree}
\usepackage{float}
\usepackage{bbold}
\usepackage{braket}
\usepackage{xcolor}
\usepackage{gensymb}

\usepackage{booktabs}
\usepackage{upgreek, eurosym}
\usepackage[colorlinks=true]{hyperref}

\newcommand{\abs}[1]{\left|#1\right|}

\graphicspath{{images/}}

\begin{document}

%-------------------------------------------------------

\bibliographystyle{apsrev}

%-------------------------------------------------------

\title{Quantum weak coin flipping with a single photon}

%-------------------------------------------------------

\author{Mathieu Bozzio}
\affiliation{Sorbonne Université, CNRS, LIP6, 4 Place Jussieu, F-75005 Paris, France}\affiliation{ Institut Polytechnique de Paris, T\'el\'ecom Paris, LTCI, 19 Place Marguerite Perey, 91129 Palaiseau, France}

\author{Ulysse Chabaud}
\affiliation{Sorbonne Université, CNRS, LIP6, 4 Place Jussieu, F-75005 Paris, France}

\author{Iordanis Kerenidis}
\affiliation{Universit\'e de Paris, CNRS, IRIF, 8 Place Aurélie Nemours, 75013 Paris, France}

\author{Eleni Diamanti}
\affiliation{Sorbonne Université, CNRS, LIP6, 4 Place Jussieu, F-75005 Paris, France}

%-------------------------------------------------------

\begin{abstract}

\textit{Weak coin flipping} is among the fundamental cryptographic primitives which ensure the security of modern communication networks.  It allows two mistrustful parties to remotely agree on a random bit when they favor opposite outcomes. Unlike other two-party computations, one can achieve information-theoretic security using quantum mechanics only: both parties are prevented from biasing the flip with probability higher than $1/2+\epsilon$, where $\epsilon$ is arbitrarily low. Classically, the dishonest party can always cheat with probability $1$ unless computational assumptions are used. Despite its importance, no physical implementation has been proposed for quantum weak coin flipping. Here, we present a practical protocol that requires a single photon and linear optics only. We show that it is fair and balanced even when threshold single-photon detectors are used, and reaches a bias as low as $\epsilon=1/\sqrt{2}-1/2\approx 0.207$. We further show that the protocol may display quantum advantage over a few hundred meters with state-of-the-art technology.

\end{abstract}

%-------------------------------------------------------

\maketitle

%-------------------------------------------------------

\section{Introduction}

   Modern communication networks are continuously expanding, as the number of users and available online resources increases. On a daily basis, users must inevitably trust local network nodes and transmission channels in order to perform sensitive tasks such as private data transmission, online banking, electronic voting, delegated computing and many more. A complex network can be secured by relying on a collection of simpler cryptographic \textit{primitives}, or building blocks, which are combined to guarantee overall security. Strong coin flipping (SCF) is one of such primitives, in which two parties remotely agree on a random bit such that none of the parties can bias the outcome with probability higher than $1/2+\epsilon$, where $\epsilon$ is the protocol bias. SCF is fundamental in multiparty computation \cite{GMW:STOC87}, online gaming and more general randomized consensus protocols involving leader election \cite{AAK:DC18}.

   Weak coin flipping (WCF) is a version of coin flipping in which both parties have a preferred, opposite outcome: it effectively designates a winner and a loser. In the classical world, information-theoretically secure SCF and WCF are impossible: they require computational assumptions or trusting a third party \cite{B:ACM83,C:STOC86,A:STOC04,BBB:PRA09}. Using quantum properties, on the other hand, enables information-theoretically secure SCF and WCF:  the lowest possible bias for quantum SCF is $\epsilon = 1/\sqrt{2}-1/2$ \cite{K:PC03}, while quantum WCF may achieve a bias arbitrarily close to zero \cite{M:arx07,ACG:SIAM16}. Remarkably, quantum WCF is also involved in the construction of optimal quantum SCF and quantum bit commitment schemes \cite{IEEE:CK09,IEEE:CK11}. Although the works from \cite{M:arx07,ACG:SIAM16} proved the existence of quantum WCF protocols achieving arbitrarily low biases, no explicit protocol had been provided. In 2002, two explicit protocols with small biases were proposed: the work from \cite{KN:IPL04} achieved $\epsilon\approx0.239$, while \cite{SR:PRL02} achieved $\epsilon=1/\sqrt{2}-1/2\approx0.207$, which is incidentally the SCF lower bound. Later, it was shown that the scheme from \cite{SR:PRL02} in fact belonged to a larger family of WCF protocols with $\epsilon=1/6\approx0.167$ \cite{M:FOCS04,M:PRA05}. Very recently, a new explicit family of protocols achieved $\epsilon\approx1/10$ \cite{ARW:arx18}, followed by arbitrarily low biases \cite{ARV:arx19}.

   While quantum SCF protocols have been experimentally demonstrated \cite{MTV:PRL05,BBB:NC11,PJL:NC14}, no implementation has been proposed for quantum WCF. This may be explained by two reasons. First, it is difficult to find an encoding and implementation which is robust to losses: a dishonest party may always declare an abort when they are not satisfied with the flip's outcome. Second, none of the previously mentioned protocols translate trivially into a simple experiment: they involve performing single-shot generalized measurements \cite{SR:PRL02} or generating beyond-qubit states \cite{KN:IPL04}.

   In this work, we introduce a family of quantum WCF protocols, inspired by \cite{SR:PRL02}, which achieve biases as low as $\epsilon = 1/\sqrt{2}-1/2\approx 0.207$. These protocols involve simple projective measurements instead of generalized ones, require a single photon and linear optics only, and need at most three rounds of communication between the parties. The information is encoded by mixing the single photon with vacuum on an unbalanced beam splitter, which generates entanglement between the photon number modes \cite{MBH:PRL13}: both parties may then agree on a random bit, while the entanglement is simultaneously verified. We also use a version of our schemes to construct a quantum SCF protocol with bias $\approx 0.31$. We further derive a practical security proof for both number-resolving and threshold single-photon detectors, considering the extension to infinite Hilbert spaces. Since the presence of losses may enable classical protocols to reach lower cheating probabilities than quantum protocols, we finally show that our fair and balanced quantum protocol bears no classical equivalent over a few hundred meters of lossy optical fiber and non-unit detection efficiency. 

%-------------------------------------------------------

\section{Protocol and correctness} In the honest protocol, Alice and Bob wish to toss a fair coin, with a priori knowledge that they each favor opposite outcomes. Fig.~\ref{fig:protocol} represents the implementation of the honest protocol, which follows five distinct steps. Defining $x\in[0,\frac12]$ as a free protocol parameter, these read:
\begin{itemize}
\item Alice mixes a single photon with the vacuum on a beam splitter of reflectivity $x$.
\item Alice keeps the first spatial mode, and directs the second spatial mode to Bob.
\item Bob mixes the state he receives with the vacuum on a beam splitter of reflectivity $y=1-\frac1{2(1-x)}$.
\item Bob measures the second register of his state with a single-photon detector, and broadcasts the outcome $c\in\{0,1\}$.
\item The last step is a verification step, which splits into two cases. If $c=0$, Alice sends her state to Bob, who mixes it with his state on a beam splitter of reflectivity $z=2x$. He then measures the two output modes with single-photon detectors. He declares Alice the winner if the outcome $(b_1,b_0)=(1,0)$ is obtained. If $c=1$: Bob discards his state, and Alice measures her state with a single-photon detector. She declares Bob the winner if the outcome is $a=0$.
\end{itemize}
\begin{figure}
\begin{center}
\includegraphics[width=1\columnwidth]{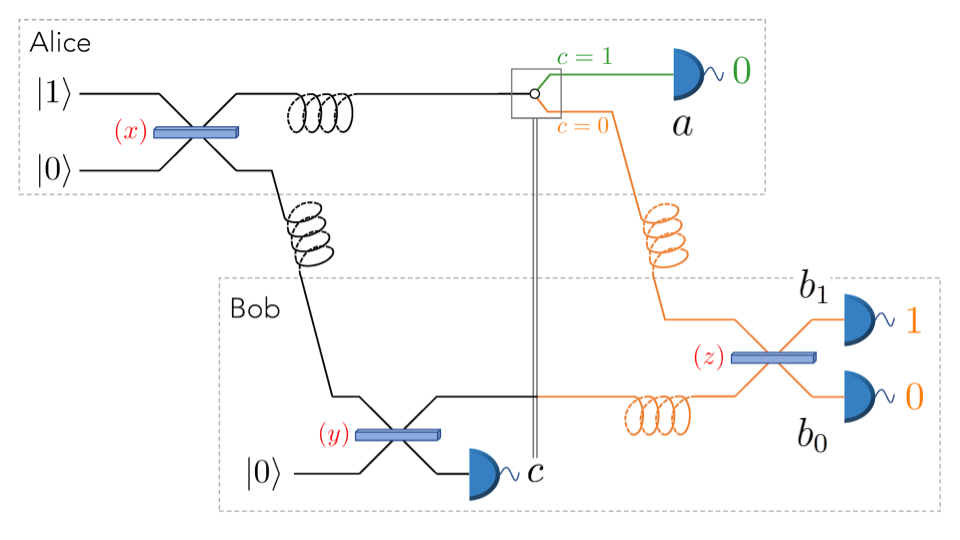}
\caption{\textbf{Representation of the honest protocol.} The dashed boxes indicate Alice and Bob's laboratories, respectively. The beam splitter reflectivities are indicated in red brackets. $\ket{0}$ and $\ket{1}$ are the vacuum and single photon Fock states, respectively. Curly lines represent fiber used for quantum communication from Alice to Bob, or delay lines within Alice's or Bob's laboratory, when waiting for the other party's communication. Bob broadcasts the classical outcome $c$, which controls an optical switch on Alice's side. The protocol when Bob declares $c=0/1$ is represented in orange/green. The final outcomes $(a,b_1,b_0)=(0,1,0)$ are the expected outcomes when both parties are honest.}
\label{fig:protocol}
\end{center}
\end{figure}
We show that the protocol is fair, i.e.\@ that the probability of winning for each party is $\frac12$ when they are both honest.

Single photons are quantized excitations of the electromagnetic field, which are described by the action of the creation operator onto the vacuum. Beam splitters act linearly on creation operators, and leave invariant the vacuum. Hence, the evolution of the quantum state over the three modes up to Bob's measurement reads:
\be
\ba
\ket{100}&\underset{(x),12}\to\sqrt{x}\ket{100}+\sqrt{1-x}\ket{010}\\
&\underset{(y),23}\to\sqrt{x}\ket{100}+\sqrt{(1-x)y}\ket{010}\\
&\quad\quad\quad\quad+\sqrt{(1-x)(1-y)}\ket{001},
\ea\label{eq:evolution2}
\ee
where the notation $(t),kl$ indicates the reflectivity of the beam splitter and the corresponding spatial modes.
Hence, the probability that Bob obtains outcome $c=1$ when measuring the third register is $P(1)=(1-x)(1-y)$, while the probability of outcome $c=0$ is $P(0)=1-P(1)$. Having set $y=1-\frac1{2(1-x)}$ ensures $P(0)=P(1)=\frac12$. When $c=1$, the state on modes $1$ and $2$ is projected onto $\ket{00}$, while $c=0$ projects the state onto $\sqrt{2x}\ket{10}+\sqrt{1-2x}\ket{01}$. In the first case, the measurement performed by Alice outputs $a=0$ with probability $1$. In the second case, the measurement performed by Bob after the beam splitter with reflectivity $z$ outputs $(b_1,b_0)=(1,0)$ with probability $1$. Hence, the probability that Alice (resp.\@ Bob) wins is directly given by $P_h^{(A)} = P(0)$ (resp.\@ $P_h^{(B)} = P(1)$). This shows that the protocol is fair, since $P(0)=P(1)=\frac12$.

%-------------------------------------------------------

\section{Security} We now derive the security of the protocol. Namely, we obtain the probabilities of winning when Bob is dishonest and Alice is honest, and vice versa.

\subsection{Dishonest Bob, honest Alice}

Dishonest Bob should always declare the outcome $c=1$ in order to maximize his winning probability. The outcome of the coin flip is then confirmed if Alice obtains the outcome $a=0$ upon verification. Bob thus needs to maximize the probability of the outcome $a=0$, applying a general quantum operation to his half of the state. However, the probability that the detector clicks is independent of Bob's action. It is given by $x$, so that Bob's winning probability is upper bounded by $(1-x)$. This upper bound is reached if Bob discards his half of the state and broadcasts $c=1$. Then, Bob's optimal cheating probability is $P_d^{(B)}=1-x$.

\subsection{Dishonest Alice, honest Bob}

 Alice wins when Bob declares $c=0$ and the outcome of his quantum measurement is $(b_1,b_0)=(1,0)$. The most general strategy of dishonest Alice is to send a (mixed) state $\sigma$, while Bob performs the rest of the protocol honestly. We find that the security is derived easily if Bob is allowed photon number resolving detectors (see Appendix \ref{sec:secu} for details of all the proofs). 
 
 Remarkably, the protocol is still secure even when Bob only uses threshold detectors, which is essential to the practicality of the protocol. Moreover, Alice's optimal cheating probability remains the same in both cases: $P_d^{(A)}=1-(1-y)(1-z)$, which equals $\frac{1}{2(1-x)}$ for $y=1-\frac{1}{2(1-x)}$ and $z=2x$. In particular, for all values of $x$, we retrieve the property shared by the protocols of~\cite{SR:PRL02}: $P_d^{(A)}P_d^{(B)}=\frac12$.
 
 The underlying idea in the security analysis for threshold and number resolving detectors is that Alice must generate the state which maximizes the overlap with Bob's projectors $\ket{100}\bra{100}$ and $\sum_{n=1}^{\infty}\ket{n00}\bra{n00}$, respectively. Setting $x=1-1/\sqrt2$, we obtain a version of the protocol which is balanced, i.e.\@ both players have the same cheating probability $1/\sqrt2$. The protocol bias is then $\epsilon=1/\sqrt2-1/2\approx0.207$.

Moreover, following~\cite{IEEE:CK09}, we show in Appendix \ref{sec:SCF} that a suitable choice of parameters $x$,$y$,$z$ yielding an unbalanced quantum WCF protocol allows to construct a quantum SCF protocol with bias $\approx0.31$.

%-------------------------------------------------------

\section{Fault tolerance}
\subsection{Noise} 
We now investigate how imperfect state generation, non-ideal beam splitters and single-photon detector dark counts affect the correctness and security of the protocol. While we fixed the parameter values to $y=1-\frac1{2(1-x)}$ and $z=2x$ in the ideal setting, we now allow the three parameters $x$, $y$, $z$ to vary freely.

The vacuum/single-photon encoding is very robust to noise, in comparison to polarization or phase encoding for instance: the only property that must be preserved through propagation is photon number. Alice may simply produce a heralded single photon via spontaneous parametric down-conversion (SPDC) \cite{C:CP18}, which generates a photon pair: one may be used for the flip, while the other may herald the presence of the first one. Given the photon-pair emission probability $p$, accidentally emitting two pairs at the same time using SPDC occurs with probability $p^2$. Since $p$ may be arbitrarily tuned by changing the pump power, $p^2$---and therefore the probability of two photons being accidentally generated by Alice at once---may then be decreased to negligible values.

Note that, in the case where Alice's single photon source is probabilistic but heralded (as in SPDC), she may always inform Bob of a successful state generation prior to his announcement of $c$ without compromising security. In what follows, we may therefore assume that both parties have agreed on the presence of an initial state, and hence know when the protocol occurs.

Noise will therefore stem from the non-ideal reflectivities of the beam splitters, and the non-zero detector dark count probability $p_{dc}$. For each party, these may affect the protocol correctness in two ways: an undesired bias of the flip, and an added abort probability during the verification process.

Deviations on the beam splitter reflectivities $x$, $y$, and $z$ will first change the honest winning probabilities: these may be re-calculated by replacing the ideal reflectivity $r\in\{x,y\}$ with an imperfect $r'$. As regards to honest aborts, a beam splitter with reflectivity $z'$ instead of $z$ may be applied on the resulting state when $c=0$. Noisy detectors may cause an unwanted abort corresponding to a click because of dark counts. However, with superconducting nanowire single-photon detectors, this probability is typically very low, of the order of $p_{dc}<10^{-8}$ \cite{H:NP09}. 
%assuming a dark count rate of $100$ Hz and timing jitter of a few hundred $ps$.

We can therefore conclude that any source of noise may be incorporated in the security analysis by simply replacing parameters $x$, $y$, and $z$ with $x'$, $y'$, and $z'$. Furthermore, this source of error will most likely be negligible with current technology. We therefore solely focus on the more consequential effects of losses.

\subsection{Losses} Losses can be due to non-unit channel and delay line transmissions, as well as non-unit detection efficiency. We label $\eta_t$ the transmission efficiency of the quantum channel from Alice to Bob. We also define as $\eta_f^{(i)}$ the transmission of party $i$'s fiber delay, while $\eta_d^{(i)}$ denotes the detection efficiency of party $i$'s single-photon detectors. Here, we assume the efficiencies of Bob's detectors to be the same, and that each party introduces a fiber delay whenever they are waiting for the other party's communication. The delay time therefore depends on the distance between the two parties. %When both parties are honest, the outcome of the flip may be reduced to the value of the declared outcome $c$, provided that the quantum verification step yields the awaited outcome: $(1,0)$ for Alice and $(0)$ for Bob.

In the presence of losses, the protocol may also abort when both parties are honest, when the photon is lost. We derive in Appendix~\ref{sec:Phlossy} the expressions for the two honest winning probabilities $P_h^{(A)}$ and $P_h^{(B)}$, and hence the probability $P_{ab}$ of abort, in the presence of losses:
\begin{equation}
    \begin{aligned}
    P_h^{(A)} &=\eta_t\eta_d^{(B)}\left(\sqrt{xz\eta_f^{(A)}}+\sqrt{(1-x)y(1-z)\eta_f^{(B)}}\right)^2\\
    P_h^{(B)} &=\eta_t\eta_d^{(B)}(1-x)(1-y)\\
    P_{ab}& = 1-P_h^{(A)}-P_h^{(B)}.
    \end{aligned}
\label{eq:loss_correct}
\end{equation}
Note that the overall correctness does not depend on Alice's detection efficiency $\eta_d^{(A)}$, since the declaration of outcome $c$ depends solely on Bob's detector, and the verification step on Alice's side involves detecting vacuum.

%-------------------------------------------------------

\section{Security in the presence of losses} Dishonest Bob's best strategy is to perform the same attack as in the lossless case, because he has no control over Alice's half of the subsystem. His winning probability is then $P_d^{(B)}=1-x\eta_f^{(A)}\eta_d^{(A)}$. However, in a more general game-theoretic scenario, Bob's best strategy will in fact depend on the rewards and sanctions associated with honest aborts and "getting caught cheating" aborts. In other words, Bob has to minimize his risk-to-reward ratio. Maximizing his winning probability makes him run the risk of getting caught cheating with probability $ x\eta_f^{(A)}\eta_d^{(A)}$.

Dishonest Alice must still generate the state which maximizes the $(b_1,b_0,c)=(1,0,0)$ outcome on Bob's detectors after his honest transformations have been applied. However, the expression for Bob's corresponding projector now changes, as there is a finite probability $(1-\eta_d^{(B)})^n$ that the $n$-photon component is projected onto the vacuum. The $0$ outcome on one spatial mode is therefore triggered by the projection $\Pi_0=\sum_{n=0}^{\infty}(1-\eta_d^{(B)})^n\ket{n}\bra{n}$. The total projector responsible for the $(b_1,b_0,c)=(1,0,0)$ outcome then reads $\Pi_{100}=\left(\mathbb{1}-\Pi_0\right)\otimes\Pi_0\otimes\Pi_0$. We show in Appendix~\ref{sec:seculoss} that dishonest Alice's maximum winning probability $P_d^{(A)}$ satisfies:
%
%\begin{widetext}
%
\be
\ba
\max_{l>0}&\left[\left(1-(1-y\eta_f^{(B)})(1-z)\eta_d^{(B)}\right)^l-\left(1-\eta_d^{(B)}\right)^l\right]\\
&\leqslant 1-(1-y)(1-z).
\ea
\ee
%
%\end{widetext}
%

The value of the upper bound on the right hand side is Alice's cheating probability in the lossless case. This shows that Alice cannot take advantage of Bob's imperfect detectors or his lossy delay line in order to increase her cheating probability. We now provide a sketch of the proof: since passive linear optical elements act linearly on creation operators, equal losses on different modes may be commuted through the interferometer of the protocol. This allows to upper bound Alice's maximum winning probability by her winning probability in an equivalent picture in which the losses happen just after her state preparation, then followed by a lossless protocol. In that case, it is as if dishonest Alice was trying to cheat in the lossless protocol, while being restricted to lossy state preparation instead of ideal state preparation.
%
%
% \begin{figure}[h]
% 	\begin{center}
% 		\includegraphics[width=\columnwidth]{B_honest_eta0.png}
% 		\caption{\textbf{Equivalent picture without a delay line.} In this picture, the losses of the detectors have been commuted back to Alice's state preparation.}
% 		\label{fig:eta0}
% 	\end{center}
% \end{figure}
%

%-------------------------------------------------------
\begin{figure}
	\begin{center}
		\includegraphics[width=\columnwidth]{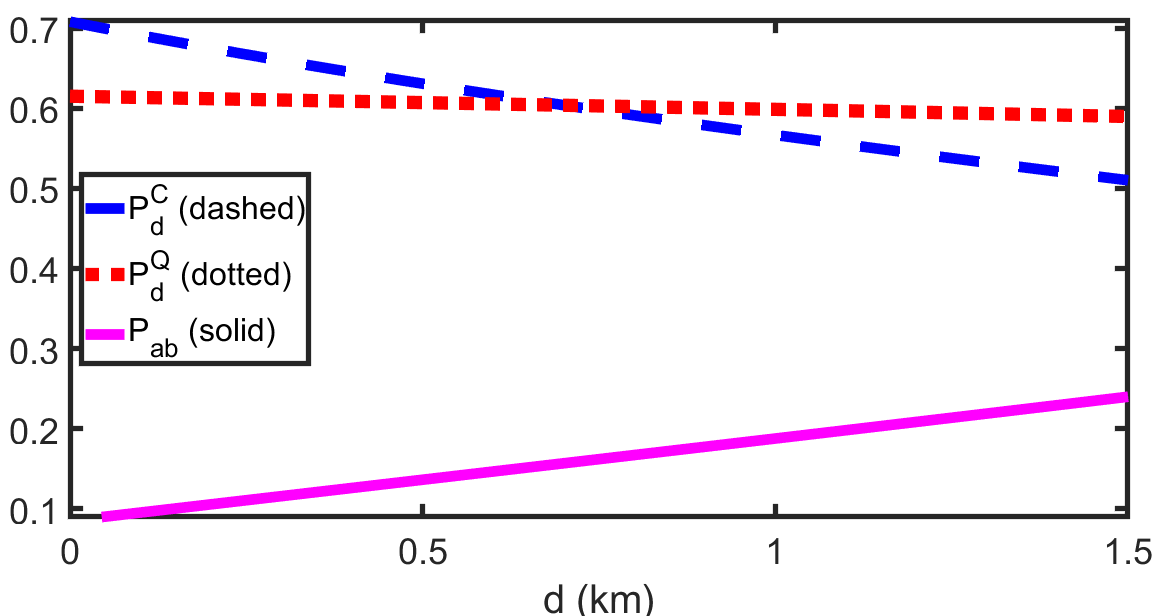}
		\caption{\textbf{Practical quantum advantage for a fair and balanced protocol.} Numerical values for the lowest classical and quantum cheating probabilities, $P_d^{C}$ and $P_d^{Q}$, are plotted as a function of distance $d$ in dashed blue and dotted red, respectively. Honest abort probability $P_{ab}$ (responsible for $P_d^{Q}$ being lower than our ideal quantum cheating probability $1/\sqrt{2}$) is plotted in solid magenta. Our quantum protocol performs strictly better than any classical protocol when  $P_d^{Q}<P_d^{C}$. We set $\eta_f=\eta_s\eta_t^2$, where $\eta_s$ is the fiber delay transmission corresponding to $500$ns of optical switching time, and $\eta_t^2=\left(10^{-\frac{0.2}{10}d}\right)^2$ is the fiber delay transmission associated with travelling distance $d$ twice (once for quantum, once for classical) in single-mode fibers with attenuation $0.2$ dB/km. We have $\eta_d=0.95$ and $z=0.57$. For performance with lower $\eta_d=0.90$, please see Appendix \ref{sec:eta}.}
		\label{fig:performance}
	\end{center}
\end{figure}

\section{Practical protocol performance} We now analyze the performance of our protocol in a practical setting, by enforcing three conditions on the free parameters: the protocol must be fair, balanced, and perform strictly better than any classical protocol. The latter condition is not required in an ideal implementation, since quantum WCF always provides a security advantage over classical WCF. Allowing for abort cases, however, may enable some classical protocols to perform better than quantum ones. This is because increasing the abort probability effectively decreases Alice and Bob's cheating probabilities. We say that the protocol allows for quantum advantage when it provides a strictly lower cheating probability than any classical protocol with the same abort probability. This is obtained using the bounds from~\cite{HW:TCC11}, which yield the best classical cheating probability $P_d^{C}=1-\sqrt{P_{ab}}$ for our protocol (see Appendix \ref{sec:system}). The three conditions may then be translated into the following system of equations, where we define $P_d^{Q}=P_d^{(A)}=P_d^{(B)}$:

\begin{equation}
   \left\{
    \begin{array}{ll}
     (i) \:\:\:\: &P_h^{(A)}=P_h^{(B)}\:\:\:\:\:\: \text{fairness} \\
     (ii) \:\:\:\: &P_d^{(A)}=P_d^{(B)}\:\:\:\:\: \text{balance} \\
     (iii) &P_d^{Q}< P_d^{C} \:\:\text{quantum advantage} \\
    \end{array}
\right.
\label{eq:systemz}
\end{equation}
%
%We can numerically obtain a range of parameters for which system (\ref{eq:systemz}) is satisfied. 
%In particular, when both parties have the same detection efficiency, condition $(iii)$ is always satisfied for $x<1$.
Fig.~\ref{fig:performance} shows a choice of parameters for which system (\ref{eq:systemz}) is satisfied, up to a distance of $d$ km.

%-------------------------------------------------------

\section{Discussion} By noticing a non-trivial connection between the early protocol from \cite{SR:PRL02} and linear optical transformations, we answer the question of the implementability of quantum weak coin flipping, and show that it is achievable with current technology over a few hundred meters. Both parties require a set of beam splitters and single photon threshold detectors. State generation on Alice's side can be performed with any heralded probabilistic single-photon source. Only Alice requires an optical switch, which is commercially available. Although short-term quantum storage is needed, a spool of optical fiber with twice the length of the quantum channel suffices, and provides the required storage/retrieval efficiency.

As the distance increases, the issue of interferometric stability should also be considered. Prior to the protocol, Alice and Bob may use similar techniques to twin-field quantum key distribution implementations to lock the interference  \cite{MPR:NP19,Pan:PRL2020}, as it is in their interest to collaborate on this task to avoid the protocol from aborting. 

On the fundamental level, our results also raise the question of a potentially deeper connection between the large family of protocols from \cite{M:FOCS04,M:PRA05,M:arx07}---which achieves biases as low as $1/6$---and linear optics. Recalling that the protocol from \cite{SR:PRL02}, and hence our protocol, is conjectured optimal for this family, its extension to many rounds should be necessary in order to lower the bias. The optimality of the one-round protocol is crucial, as a recent result shows that the WCF bias decreases very inefficiently with the number of rounds \cite{M:ACM20}.

%-------------------------------------------------------

\section*{Acknowledgments} We thank Atul Singh Arora and Simon Neves for useful discussions on quantum weak coin flipping and on experimental requirements for heralded single-photon sources, respectively. We acknowledge support of the European Union’s Horizon 2020 Research and Innovation Programme under Grant Agreement No. 820445 (QIA) and the ANR through the ANR-17-CE39-0005 (quBIC) project.

\makeatletter

%\newpage

%-------------------------------------------------------

%-------------------------------------------------------

%\begin{widetext}

%\lipsum[1-3]

\onecolumngrid
%\vspace{\columnsep}
%\lipsum[5-6]
%\vspace{\columnsep}
%-------------------------------------------------------
\bigskip

\bigskip

\bigskip

\begin{center}

\Large{\textbf{Appendix}}

\end{center}

\noindent In this appendix, we give detailed proofs of the results presented in the main text. Section~\ref{sec:preli} contains preliminary technical results. In Section~\ref{sec:secu}, we provide the security analysis for dishonest Alice. In Section~\ref{sec:SCF}, we show how an unbalanced version of our WCF protocol may yield a SCF protocol. In Section~\ref{sec:Phlossy} we derive the completeness of the protocol in the lossy case, and in Section~\ref{sec:seculoss}, we extend the security analysis to the case of a lossy protocol. In Section~\ref{sec:system} we solve the system in Eq.~(\ref{eq:systemz}) of the main text, and derive the constraints that the parameters of the protocol must satisfy in order to obtain a fair and balanced protocol which still outperforms all classical WCF protocols in the lossy case. Finally, in Section~\ref{sec:eta}, we display the practical performance of our fair and balanced protocol for various detection efficiencies.

\begin{appendix}

%-------------------------------------------------------

\section{Preliminary results}
\label{sec:preli}

Single photons are obtained by the action of the creation operator onto the vacuum. Beam splitters act linearly on creation operators, and leave invariant the vacuum. More precisely, a beam splitter of reflectivity $t$ acting on modes $k,l$ maps the creation operators $\hat a_k^\dag,\hat a_l^\dag$ onto $\hat b_k^\dag,\hat b_l^\dag$, where
\be
\begin{pmatrix}\hat b_k^\dag\\\hat b_l^\dag\end{pmatrix}=H_{kl}^{(t)}\begin{pmatrix}\hat a_k^\dag\\\hat a_l^\dag\end{pmatrix},
\ee
where
\be
H_{kl}^{(r)}=\begin{pmatrix}\sqrt r&\sqrt{1-r}\\\sqrt{1-r} & -\sqrt r\end{pmatrix}.
\ee

In the following, we make use of a simple reduction which allows to simplify calculations in the proofs:

\begin{lem} Let $U=(H^{(z)}\otimes \mathbb{1})(\mathbb{1}\otimes H^{(y)})$, with $z>0$. For all density matrices $\tau$,
\begin{equation}
\mathrm{Tr}[(\tau\otimes\ket0\bra0)U^\dag(\mathbb{1}\otimes\ket{00}\bra{00})U]=\mathrm{Tr}[(\tau\otimes\ket0\bra0)V^\dag(\ket{0}\bra{0}\otimes \mathbb{1}\otimes\ket{0}\bra{0})V],
\end{equation}
where $V=(\mathbb{1}\otimes H^{(b)})(H^{(a)}\otimes \mathbb{1})(\mathbb{1}\otimes R(\pi)\otimes \mathbb{1})$, with $a=\frac{y(1-z)}{1-(1-y)(1-z)}$ and $b=1-(1-y)(1-z)$, and $R(\pi)$ a phase shift of $\pi$ acting on mode $2$.
\label{lem:reductionUV}
\end{lem}

\begin{proof}

\noindent The action of $U$ on the creation operators is given by
\be
\ba
U&=\begin{pmatrix}\sqrt z&\sqrt{1-z}&0\\\sqrt{1-z}&-\sqrt z&0\\0&0&1\end{pmatrix}\begin{pmatrix}1&0&0\\0&\sqrt y&\sqrt{1-y}\\0&\sqrt{1-y}&-\sqrt y\end{pmatrix}\\
&=\begin{pmatrix}\sqrt z&\sqrt{y(1-z)}&\sqrt{(1-y)(1-z)}\\\sqrt{1-z}&-\sqrt{yz}&-\sqrt{(1-y)z}\\0&\sqrt{1-y}&-\sqrt y\end{pmatrix}.
\ea
\ee
Linear interferometers map product coherent states onto product coherent states, and, for all $\alpha\in\mathbb C$, we have that $U^\dag\ket{\alpha00}=\ket{\beta_1\beta_2\beta_3}$, where
\be
\begin{pmatrix}\beta_1\\\beta_2\\\beta_3\end{pmatrix}=\begin{pmatrix}\alpha\sqrt z\\\alpha\sqrt{y(1-z)}\\\alpha\sqrt{(1-y)(1-z)}\end{pmatrix}.
\ee
We have $V=(\mathbb{1}\otimes H^{(b)})(H^{(a)}\otimes \mathbb{1})(\mathbb{1}\otimes R(\pi)\otimes \mathbb{1})$, with $a,b\in[0,1]$, and $R(\pi)$ a phase shift of $\pi$ acting on mode $2$. The action of $V$ on the creation operators is given by
\be
\ba
V&=\begin{pmatrix}1&0&0\\0&\sqrt b&\sqrt{1-b}\\0&\sqrt{1-b}&-\sqrt b\end{pmatrix}\begin{pmatrix}\sqrt a&\sqrt{1-a}&0\\\sqrt{1-a}&-\sqrt a&0\\0&0&1\end{pmatrix}\begin{pmatrix}1&0&0\\0&-1&0\\0&0&1\end{pmatrix}\\
&=\begin{pmatrix}\sqrt a&-\sqrt{1-a}&0\\\sqrt{b(1-a)}&\sqrt{ab}&\sqrt{1-b}\\\sqrt{(1-a)(1-b)}&\sqrt{a(1-b)}&-\sqrt b\end{pmatrix}.
\ea
\ee
For all $\alpha\in\mathbb C$, $V^\dag\ket{0\alpha0}=\ket{\gamma_1\gamma_2\gamma_3}$, where
\be
\begin{pmatrix}\gamma_1\\\gamma_2\\\gamma_3\end{pmatrix}=\begin{pmatrix}\alpha\sqrt{b(1-a)}\\\alpha\sqrt{ab}\\\alpha\sqrt{1-b}\end{pmatrix}.
\ee
Since $a=\frac{y(1-z)}{1-(1-y)(1-z)}$ and $b=1-(1-y)(1-z)$, we have $b(1-a)=z$, $ab=y(1-z)$, and $1-b=(1-y)(1-z)$, so $(\beta_1,\beta_2,\beta_3)=(\gamma_1,\gamma_2,\gamma_3)$.
Then,
\be
\ba
\mathrm{Tr}[(\tau\otimes\ket0\bra0)U^\dag(\mathbb{1}\otimes\ket{00}\bra{00})U]&=\frac1\pi\int_{\mathbb C}{d^2\alpha\mathrm{Tr}[(\tau\otimes\ket0\bra0)U^\dag\ket{\alpha00}\bra{\alpha00}U]}\\
&=\frac1\pi\int_{\mathbb C}{d^2\alpha\mathrm{Tr}[(\tau\otimes\ket0\bra0)V^\dag\ket{0\alpha0}\bra{0\alpha0}V]}\\
&=\mathrm{Tr}[(\tau\otimes\ket0\bra0)V^\dag(\ket{0}\bra{0}\otimes \mathbb{1}\otimes\ket{0}\bra{0})V],
\ea
\ee
where we used the completeness relation of coherent states $\mathbb{1}=\frac1\pi\int_{\mathbb C}{\ket\alpha\bra\alpha d^2\alpha}$.

\end{proof}

We also recall a useful simple property, which we will use extensively in the following:

\begin{lem} Equal losses on various modes can be commuted through passive linear optical elements acting on these modes.
\label{lem:commut}
\end{lem}

 This result was proven, e.g., in~\cite{berry2010linear}, and we give hereafter a quick proof for completeness.

\begin{proof}

One way to prove this statement is to use the fact that any interferometer may be decomposed as beam splitters and phase shifters~\cite{reck1994experimental}. Then, losses trivially commute with phase shifters, and are easily shown to commute with beam splitters. Indeed, consider a beam splitter of reflectivity $t$ acting on modes $1$ and $2$. Its action on the creation operators of the modes is given by
\be
\hat a_1^\dag,\hat a_2^\dag\rightarrow\sqrt t\hat a_1^\dag+\sqrt{1-t}\hat a_2^\dag,\sqrt{1-t}\hat a_1^\dag-\sqrt t\hat a_2^\dag,
\ee
while equal losses $\eta$ on both modes act as
\be
\hat a_1^\dag,\hat a_2^\dag\rightarrow\sqrt\eta \hat a_1^\dag,\sqrt\eta \hat a_2^\dag.
\ee
Hence, the action of the beam splitter followed by losses is given by
\be
\hat a_1^\dag,\hat a_2^\dag\rightarrow\sqrt\eta(\sqrt t\hat a_1^\dag+\sqrt{1-t}\hat a_2^\dag),\sqrt\eta(\sqrt{1-t}\hat a_1^\dag-\sqrt t\hat a_2^\dag),
\ee
while losses followed by the beam splitter act as
\be
\hat a_1^\dag,\hat a_2^\dag\rightarrow\sqrt t(\sqrt\eta \hat a_1^\dag)+\sqrt{1-t}(\sqrt\eta \hat a_2^\dag),\sqrt{1-t}(\sqrt\eta \hat a_1^\dag)-\sqrt t(\sqrt\eta \hat a_2^\dag),
\ee
which is equal to the previous evolution.

\end{proof}

%-------------------------------------------------------

In what follows, we let the parameters $x,y,z$ vary freely, and derive the relation these parameters need to satisfy to enforce a honest protocol without abort cases.
As presented in the main text, when both parties are honest (Fig.~\ref{fig:protocol}), the evolution of the quantum state over the three modes up to Bob's first measurement reads:
\be
\ba
\ket{100}&\underset{(x),12}\to\sqrt{x}\ket{100}+\sqrt{1-x}\ket{010}\\
&\underset{(y),23}\to\sqrt{x}\ket{100}+\sqrt{(1-x)y}\ket{010}+\sqrt{(1-x)(1-y)}\ket{001},
\ea\label{eq:evolution}
\ee
where the notation $(t),kl$ indicates the reflectivity of the beam splitter and the corresponding spatial modes.
Hence, the probability that Bob obtains outcome $c=1$ when measuring the third register is $P(1)=(1-x)(1-y)$, while the probability of outcome $c=0$ is $P(0)=1-P(1)$.

\medskip

If Bob registers the outcome $c=1$, then the post-measurement state on Alice's side is $\ket0$, which will always pass the verification step. 

\medskip

If Bob registers the outcome $c=0$, then the post-measurement state reads:
\be
\sqrt{\frac{x}{1-(1-x)(1-y)}}\ket{10}+\sqrt{\frac{(1-x)y}{1-(1-x)(1-y)}}\ket{01}.
\ee
The value of the parameter $z$ should be fixed to
\be
z=\frac{x}{1-(1-x)(1-y)},
\label{eq:z}
\ee
so that this state passes the verification step, and that the protocol doesn't abort in the honest case. We assume this relation holds in the following. In that case, the winning probabilities of Alice and Bob in the honest case are given by
\be
\ba
P_h^{(A)}&=1-(1-x)(1-y)\\
P_h^{(B)}&=(1-x)(1-y).
\ea
\label{ph1}
\ee
The protocol is fair when $(1-x)(1-y)=\frac12$. In that case, $y=1-\frac1{2(1-x)}$ and $z=2x$. 

Let us also recall from the main text that in the general case, the winning probability of dishonest Bob is given by
\be
P_d^{(B)}=1-x.
\label{pdb1}
\ee

\section{Security analysis for Dishonest Alice without losses}
\label{sec:secu}

\subsection{Bob has number-resolving detectors}

When using number-resolving single-photon detectors, any projection onto the $n>1$ photon subspace leads to Alice getting caught cheating. Alice must therefore maximize the overlap with the projective measurement $\ket{100}\bra{100}$ only (Fig.~\ref{fig:Bhonest}).
\begin{figure}[h!]
	\begin{center}
		\includegraphics[width=0.6\columnwidth]{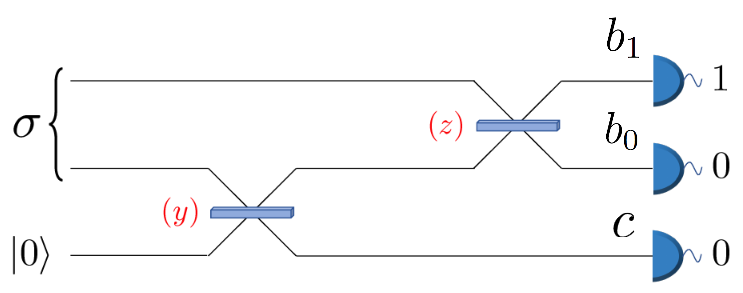}
		\caption{\textbf{Dishonest Alice.} Alice aims to maximize the outcome $(1,0,0)$: an outcome $0$ on the third mode means that Bob declared Alice the winner, while an outcome $(1,0)$ for modes $1$ and $2$ means that Alice passed Bob's verification. The reflectivities of the beamsplitter are given by $y=1-\frac1{2(1-x)}$ and $z=2x$.}
		\label{fig:Bhonest}
	\end{center}
\end{figure}
Let $\sigma$ be the state sent by Alice. Let $U=(H^{(z)}\otimes \mathbb{1})(\mathbb{1}\otimes H^{(y)})$, with $z=\frac{x}{1-(1-x)(1-y)}$.
Alice needs to maximize the probability of the overall outcome $(1,0,0)$, which is given by
\begin{equation}
    P_d^{(A)}=\mathrm{Tr}[U(\sigma\otimes\ket0\bra0)U^\dag\ket{100}\bra{100}],
\end{equation}
since Bob uses number-resolving detectors. By convexity of the probabilities, we may assume without loss of generality that Alice sends a pure state $\sigma=\ket\psi\bra\psi$, which allows us to write:
\be
\ba
P_d^{(A)}&=\mathrm{Tr}[U(\ket\psi\bra\psi\otimes\ket0\bra0)U^\dag\ket{100}\bra{100}]\\
&=\mathrm{Tr}[(\ket\psi\bra\psi\otimes\ket0\bra0)U^\dag\ket{100}\bra{100}U]\\
&=\mathrm{Tr}[\bra{\psi}\otimes\bra{0}U^\dag\ket{100}\bra{100}U\ket{\psi}\otimes\ket{0}].
\ea\label{pwin}
\ee
We have:
\be
\ba
U^\dag\ket{100}&=(\mathbb{1}\otimes H^{(y)})(H^{(z)}\otimes \mathbb{1})\ket{100}\\
&=(\mathbb{1}\otimes H^{(y)})(\sqrt{z}\ket{100}+\sqrt{1-z}\ket{010})\\
&=\sqrt{z}\ket{100}+\sqrt{y(1-z)}\ket{010}+\sqrt{(1-y)(1-z)}\ket{001},
\ea
\ee
and therefore:
\be
\ba
U^\dag\ket{100}\bra{100}U=&z\ket{100}\bra{100}+y(1-z)\ket{010}\bra{010}+(1-y)(1-z)\ket{001}\bra{001}\\
&+\sqrt{yz(1-z)}\left(\ket{100}\bra{010}+\ket{010}\bra{100}\right)\\
&+\sqrt{z(1-y)(1-z)}\left(\ket{100}\bra{001}+\ket{001}\bra{100}\right)\\
&+(1-z)\sqrt{y(1-y)}\left(\ket{010}\bra{001}+\ket{001}\bra{010}\right).
\ea
\ee
Substituting back into Eq. (\ref{pwin}) then reduces to:

\be
\ba
P_d^{(A)}&=\bra{\psi}\left(z\ket{10}\bra{10}+y(1-z)\ket{01}\bra{01}+\sqrt{yz(1-z)}(\ket{10}\bra{01}+\ket{01}\bra{10})\right)\ket{\psi}\\
&=\bra{\psi}\left(\sqrt{z}\ket{10}+\sqrt{y(1-z)}\ket{01}\right)\left(\sqrt{z}\bra{10}+\sqrt{y(1-z)}\bra{01}\right)\ket{\psi}\\
&=\abs{\bra{\psi}\left(\sqrt{z}\ket{10}+\sqrt{y(1-z)}\ket{01}\right)}^2.
\ea
\ee
Using Cauchy-Schwarz inequality then allows to upper bound $P_d^{(A)}$ as:
\begin{equation}
\begin{aligned}
    P_d^{(A)} &\leqslant \|\psi\|^2\left\|\left(\sqrt{z}\ket{10}+\sqrt{y(1-z)}\ket{01}\right)\right\|^2 \leqslant(1-(1-y)(1-z))\|\psi\|^2,
\end{aligned}
\label{cauchy}
\end{equation}
which is maximized for $\|\psi\|=1$. Hence we finally get:

\begin{equation}
P_d^{(A)} \leqslant1-(1-y)(1-z).
\end{equation}

In order to find Alice's optimal cheating strategy (i.e.\@ the optimal pure state $\ket{\phi}$ that she must send to achieve this bound), we remark that the unnormalized state $\sqrt{z}\ket{10}+\sqrt{y(1-z)}\ket{01}$ maximizes the expression in Eq. (\ref{cauchy}). Normalizing this state then provides Alice's optimal strategy, which is to prepare the state
\begin{equation}
\ket{\phi}:=\sqrt{\frac z{1-(1-y)(1-z)}}\ket{10}+\sqrt{\frac{y(1-z)}{1-(1-y)(1-z)}}\ket{01}.
\label{eq:phi}
\end{equation}
Hence,
\begin{equation}
P_d^{(A)}=1-(1-y)(1-z).
\end{equation}
In the case of a fair protocol, $y=1-\frac1{2(1-x)}$ and $z=2x$, so
\begin{equation}
P_d^{(A)}=\frac1{2(1-x)},
\end{equation}
and Alice's optimal strategy is to prepare the state
\begin{equation}
\ket{\phi_x}:=2\sqrt{x(1-x)}\ket{10}+(1-2x)\ket{01}.
\end{equation}

\subsection{Bob has threshold detectors}

Unlike the previous case, incorrect outcomes with higher photon number could still pass the test: for $n\ge1$, the threshold detectors cannot discriminate between a $\ket{100}$ and $\ket{n00}$ projection. We show in the following that this doesn't help a dishonest Alice, and that the strategy described previously for the case of number resolving detectors is still optimal in the case of threshold detectors.

With the same notations as in the previous proof, Alice needs to maximize the probability of the overall outcome $(1,0,0)$, hence the overlap with the projector $\sum_{n=1}^{\infty}\ket{n00}\bra{n00}=\left(\mathbb{1}-\ket{0}\bra{0}\right)\otimes\ket{00}\bra{00}$. This allows us to write:
\begin{equation}
    P_d^{(A)}=\mathrm{Tr}[U(\ket\psi\bra\psi\otimes\ket0\bra0)U^\dag((\mathbb{1}-\ket0\bra0)\otimes\ket{00}\bra{00})],
    \label{eq:PwinAT}
\end{equation}
since Bob uses threshold detectors, where $U=(H^{(z)}\otimes \mathbb{1})(\mathbb{1}\otimes H^{(y)})$, with $z=\frac{x}{1-(1-x)(1-y)}$.

Linear optical evolution conserves photon number. Hence if Alice sends the vacuum state, the detectors will never click. Removing the two-mode vacuum component of the state prepared by Alice and renormalizing therefore always increases her winning probability. Since we are looking for the maximum winning probability, we can assume without loss of generality that $\braket{\psi|00}=0$, i.e.
\be
\mathrm{Tr}[U(\ket\psi\bra\psi\otimes\ket0\bra0)U^\dag\ket{000}\bra{000}]=|\braket{\psi|00}|^2,
\ee
So maximizing the winning probability in Eq.~(\ref{eq:PwinAT}) is equivalent to maximizing
\be
\tilde P_d^{(A)}=\mathrm{Tr}[U(\ket\psi\bra\psi\otimes\ket0\bra0)U^\dag(\mathbb{1}\otimes\ket{00}\bra{00})],
\ee
given the constraint $\braket{\psi|00}=0$. We have
\be
\ba
\tilde P_d^{(A)}&=\mathrm{Tr}[U(\ket\psi\bra\psi\otimes\ket0\bra0)U^\dag(\mathbb{1}\otimes\ket{00}\bra{00})]\\
&=\mathrm{Tr}[(\ket\psi\bra\psi\otimes\ket0\bra0)U^\dag(\mathbb{1}\otimes\ket{00}\bra{00})U].
\label{eq:completeness}
\ea
\ee
With Lemma~\ref{lem:reductionUV} and Eq.~(\ref{eq:completeness}), we may thus write:
\be
\tilde P_d^{(A)}=\mathrm{Tr}[(\ket\psi\bra\psi\otimes\ket0\bra0)V^\dag(\ket{0}\bra{0}\otimes \mathbb{1}\otimes\ket{0}\bra{0})V],
\label{eq:PwinAtilde}
\ee
where $V=(\mathbb{1}\otimes H^{(b)})(H^{(a)}\otimes \mathbb{1})(\mathbb{1}\otimes R(\pi)\otimes \mathbb{1})$, with $a=\frac{y(1-z)}{1-(1-y)(1-z)}$ and $b=1-(1-y)(1-z)$. Let us now define:
\be
\ket{\psi_a}:=H^{(a)}(\mathbb{1}\otimes R(\pi))\ket\psi.
\label{eq:psix}
\ee
The constraints $\braket{\psi|00}=0$ and $\braket{\psi_a|00}=0$ are equivalent, because the above transformation leaves the total number of photons invariant. With Eq.~(\ref{eq:PwinAtilde}) we obtain
\be
\tilde P_d^{(A)}=\mathrm{Tr}[(\ket{\psi_a}\bra{\psi_a}\otimes\ket0\bra0)(\mathbb{1}\otimes H^{(b)})(\ket{0}\bra{0}\otimes \mathbb{1}\otimes\ket{0}\bra{0})(\mathbb{1}\otimes H^{(b)})],
\label{eq:PwinAtilde2}
\ee
with the constraint $\braket{\psi_a|00}=0$.
\begin{figure}
	\begin{center}
		\includegraphics[width=0.5\columnwidth]{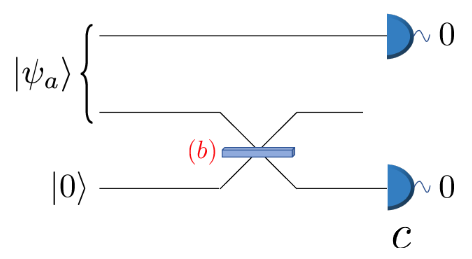}
		\caption{\textbf{Equivalent picture for dishonest Alice.} In the original dishonest setup of Fig.~\ref{fig:Bhonest}, Alice aims to maximize the outcome $(1,0,0)$. This is equivalent to Alice maximizing outcome $0$ on spatial modes $1$ and $3$, independently of what is detected on mode $2$. The outcomes indicated correspond to Alice winning. The reflectivity is $b=1-(1-y)(1-z)$.}
		\label{fig:equivproof}
	\end{center}
\end{figure}
Maximizing this expression thus corresponds to maximizing the probability of the outcome $(0,0)$ when measuring modes $1$ and $3$ of the state obtain by mixing the second half of $\ket{\psi_a}$ with the vacuum on a beam splitter of reflectivity $b=1-(1-y)(1-z)$ (Fig.~\ref{fig:equivproof}).

We now show that an optimal strategy for Alice is to ensure that $\ket{\psi_a}=\ket{01}$. Let us write
\be
\ket{\psi_a}=\sum_{p+q>0}{\psi_{pq}\ket{pq}},
\ee
where we take into account the constraint $\braket{\psi_x|00}=0$. Then, with Eq.~(\ref{eq:PwinAtilde2}) we obtain
\be
\ba
\tilde P_d^{(A)}&=\sum_{p+q>0,p'+q'>0}{\psi_{pq}\psi_{p'q'}^*\mathrm{Tr}[\ket{pq0}\bra{p'q'0}(\ket{0}\bra{0}\otimes H^{(b)}(\mathbb{1}\otimes\ket{0}\bra{0})H^{(b)})]}\\
&=\sum_{q>0,q'>0}{\psi_{0q}\psi_{0q'}^*\mathrm{Tr}[\ket{q0}\bra{q'0} H^{(b)}(\mathbb{1}\otimes\ket{0}\bra{0})H^{(b)}]}\\
&=\sum_{n\ge0,q>0,q'>0}{\psi_{0q}\psi_{0q'}^*\mathrm{Tr}[\ket{q0}\bra{q'0} H^{(b)}\ket{n0}\bra{n0}H^{(b)}]}\\
&=\sum_{n>0}{|\psi_{0n}|^2|\braket{n0|H^{(b)}|n0}|^2}\\
&=\sum_{n>0}{|\psi_{0n}|^2b^n},
\ea
\ee
where we used in the fourth line the fact that $H^{(b)}$ doesn't change the number of photons. Since $b\in[0,1]$, this shows that
\be
\ba
\tilde P_d^{(A)}&\leqslant b\sum_{n>0}{|\psi_{0n}|^2}\\
&=b,
\ea
\ee
since $\ket{\psi_a}$ is normalized, and this bound is reached for $|\psi_{01}|^2=1$, i.e.\@ $\ket{\psi_a}=\ket{01}$. With Eq.~(\ref{eq:psix}), this implies that an optimal strategy for Alice is to prepare the state
\begin{equation}
\begin{aligned}
\ket\psi&=(\mathbb{1}\otimes R(\pi))H^{(a)}\ket{01}\\
&=\sqrt{1-a}\ket{10}+\sqrt{a}\ket{01}\\
&=\sqrt{\frac z{1-(1-y)(1-z)}}\ket{10}+\sqrt{\frac{y(1-z)}{1-(1-y)(1-z)}}\ket{01}\\
&=\ket{\phi},
\end{aligned}
\end{equation}
where $\ket\phi$ is the state that dishonest Alice needs to send to maximize her winning probability when Bob uses number-resolving detectors (Eq. (\ref{eq:phi})). Her winning probability is then
\be
P_d^{(A)}=1-(1-y)(1-z). 
\label{pda1}
\ee
We therefore recover the same result as for number-resolving detectors. Once again, if the protocol is fair then $y=1-\frac1{2(1-x)}$ and $z=2x$, so
\be
P_d^{(A)}=\frac1{2(1-x)},
\ee
and an optimal strategy for Alice is to prepare the state
\be
\ket{\phi_x}:=2\sqrt{x(1-x)}\ket{10}+(1-2x)\ket{01}.
\ee
%

%-------------------------------------------------------

\section{Quantum SCF protocol}
\label{sec:SCF}

An unbalanced quantum WCF protocol can be turned into a quantum SCF protocol using an additional classical protocol, as described in~\cite{IEEE:CK09}. In particular, let us consider a WCF protocol such that:
\be
\ba
P_h^{(A)}&=p\\
P_h^{(B)}&=1-p\\
P_d^{(A)}&=p+\epsilon\\
P_d^{(B)}&=1-p+\epsilon,
\ea
\label{pepsSCF}
\ee
for $p\in[0,1]$ and $\epsilon>0$. Then, the corresponding SCF protocol has bias~\cite{IEEE:CK09}
\be
\max{\left(\frac12-\frac12(p-\epsilon),\frac1{2-(p+\epsilon)}-\frac12\right)}.
\ee
For our WCF protocol, we have Eqs.~(\ref{ph1}),(\ref{pdb1}) and (\ref{pda1}):
\be
\ba
P_h^{(A)}&=1-(1-x)(1-y)\\
P_h^{(B)}&=(1-x)(1-y)\\
P_d^{(A)}&=1-(1-y)(1-z)\\
P_d^{(B)}&=1-x,
\ea
\ee
with the constraint $z=\frac x{1-(1-x)(1-y)}$ (so that the protocol does not abort in the honest case, Eq.~(\ref{eq:z})).
Enforcing the conditions in Eq.~(\ref{pepsSCF}), and optimizing over the corresponding SCF bias implies
\be
\ba
x&=\frac{y^2}{(1-y)(1-2y)}\\
z&=\frac y{(1-y)^2}\\
1-\frac x2&=\frac1{2-y-z+yz},
\ea
\ee
which in turn give the values
\be
\ba
x&\approx0.38\\
y&\approx0.31\\
z&\approx0.66,
\ea
\ee
by enforcing $x,y,z\in[0,1]$, and a bias of $\approx0.31$, which is a lower bias than the best implemented SCF protocol so far \cite{PJL:NC14}.

%-------------------------------------------------------

\section{Correctness, with losses}
\label{sec:Phlossy}

We give a representation of the honest protocol with losses, in Fig.~\ref{fig:protocolapp}. The efficiency of Alice's and Bob's detectors are denoted $\eta_d^{(A)}$ and $\eta_d^{(B)}$, respectively. The efficiency of the quantum channel from Alice to Bob is denoted $\eta_t$, and $\eta_f^{(A)}$ and $\eta_f^{(B)}$ are the efficiencies of Alice's and Bob's fiber delay lines, respectively.

The honest winning probability for Bob is directly given by his chance of detecting the photon (the photon gets to his detector and doesn't get lost):
\be
P_h^{(B)}=\eta_t\eta_d^{(B)}(1-x)(1-y).
\label{eq:phb}
\ee
On the other hand, Alice wins if the photon, starting from her first input mode, is detected by Bob in the last step. 
\begin{figure}[h]
\begin{center}
\includegraphics[width=1\columnwidth]{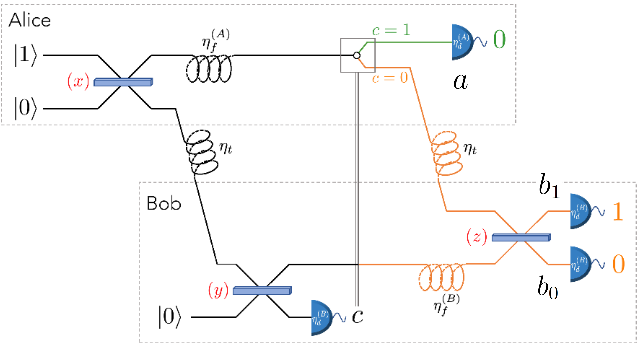}
\caption{\textbf{Representation of the honest protocol with losses.} The dashed boxes indicate Alice and Bob's laboratories, respectively. The reflectivity of the beamsplitters is indicated in red. The efficiencies of the detectors, are indicated in white. Curly lines represent fiber used for quantum communication from Alice to Bob, or delay lines within Alice's or Bob's laboratory. $\ket{0}$ and $\ket{1}$ are the vacuum and single photon Fock states, respectively. Bob broadcasts the classical outcome $c$, which controls an optical switch on Alice's side. The protocol when Bob declares $c=0/1$ is represented in orange/green. The final outcomes are the expected outcomes when both parties are honest.}
\label{fig:protocolapp}
\end{center}
\end{figure}

The evolution of the creation operator of the first mode during the lossy honest protocol is given by:
\be
\ba
\hat a_1^\dag&\rightarrow\sqrt x\hat a_1^\dag+\sqrt{1-x}\hat a_2^\dag\\
&\rightarrow\sqrt{x\eta_f^{(A)}}\hat a_1^\dag+\sqrt{(1-x)\eta_t}\hat a_2^\dag\\
&\rightarrow\sqrt{x\eta_f^{(A)}}\hat a_1^\dag+\sqrt{(1-x)\eta_ty}\hat a_2^\dag+\sqrt{(1-x)(1-y)\eta_t}\hat a_3^\dag\\
&\rightarrow\sqrt{x\eta_f^{(A)}}\hat a_1^\dag+\sqrt{(1-x)\eta_ty}\hat a_2^\dag+\sqrt{(1-x)(1-y)\eta_t\eta_d^{(B)}}\hat a_3^\dag\\
&\rightarrow\sqrt{x\eta_f^{(A)}\eta_t}\hat a_1^\dag+\sqrt{(1-x)\eta_ty\eta_f^{(B)}}\hat a_2^\dag+\sqrt{(1-x)(1-y)\eta_t\eta_d^{(B)}}\hat a_3^\dag\\
&\rightarrow\left(\sqrt{x\eta_f^{(A)}\eta_tz}+\sqrt{(1-x)\eta_ty\eta_f^{(B)}(1-z)}\right)\hat a_1^\dag+\left(\sqrt{x\eta_f^{(A)}\eta_t(1-z)}-\sqrt{(1-x)\eta_ty\eta_f^{(B)}z}\right)\hat a_2^\dag\\
&\qquad+\sqrt{(1-x)(1-y)\eta_t\eta_d^{(B)}}\hat a_3^\dag\\
&\rightarrow\left(\sqrt{x\eta_f^{(A)}\eta_tz\eta_d^{(B)}}+\sqrt{(1-x)\eta_ty\eta_f^{(B)}(1-z)\eta_d^{(B)}}\right)\hat a_1^\dag+\left(\sqrt{x\eta_f^{(A)}\eta_t(1-z)\eta_d^{(B)}}-\sqrt{(1-x)\eta_ty\eta_f^{(B)}z\eta_d^{(B)}}\right)\hat a_2^\dag\\
&\qquad+\sqrt{(1-x)(1-y)\eta_t\eta_d^{(B)}}\hat a_3^\dag.
\ea
\ee
In particular, the photon reaches Bob's uppermost detector with probability
\be
\ba
P_h^{(A)}&=\left(\sqrt{x\eta_f^{(A)}\eta_tz\eta_d^{(B)}}+\sqrt{(1-x)\eta_ty\eta_f^{(B)}(1-z)\eta_d^{(B)}}\right)^2\\
&=\eta_t\eta_d^{(B)}\left(\sqrt{xz\eta_f^{(A)}}+\sqrt{(1-x)y(1-z)\eta_f^{(B)}}\right)^2.
\ea
\label{eq:pha}
\ee
Finally, the protocol aborts for all other detection events:
\be
P_{ab}=1-P_h^{(A)}-P_h^{(B)}.
\ee
%

%-------------------------------------------------------

\section{Security analysis for Dishonest Alice, with losses}
\label{sec:seculoss}

The losses $\eta$ correspond to a probability $1-\eta$ of losing a photon. These can be modelled as a mixing with the vacuum on a beam splitter of reflectivity $\eta$. Dishonest Bob wins with probability
\be
P_d^{(B)}=1-x\eta_f^{(A)}\eta_d^{(A)},
\label{eq:pdb}
\ee
by performing the same attack as in the lossless case, since he has no control over Alice's laboratory. In what follows, we provide the security analysis for Dishonest Alice.

\subsection{Lossy delay line}

We show in this section that Alice’s maximum winning probability when Bob is using a delay line of efficiency $\eta_f$ is always lower than when Bob’s delay line is perfect, i.e.\@ $\eta_f=1$, independently of the efficiency $\eta_d$ of his detectors. The lossy delay line of efficiency $\eta_f$ may be modelled as a mixing with the vacuum on a beam splitter of transmission $\eta_f$.

Alice prepares a state $\sigma$, which goes through the interferometer depicted in Fig.~\ref{fig:lossy1}, and wins if the measurement outcome obtained by Bob is $(1,0,0)$.
In particular, note that the outcome $0$ must be obtained for the third mode. Hence Alice's winning probability is always lower than if the third mode was mixed with the vacuum on a beam splitter of transmission amplitude $\eta_f$ just before the detection (Fig.~\ref{fig:lossy2}), since this increases the probability of the outcome $0$ for this mode.
Let us assume that this is the case. Then, by Lemma~\ref{lem:commut}, the losses $\eta_f$ on output modes $2$ and $3$ may be commuted back through the beam splitter of reflectivity $y$, acting on modes $2$ and $3$. 

\begin{figure}[h!]
	\begin{center}
		\includegraphics[width=4.3in]{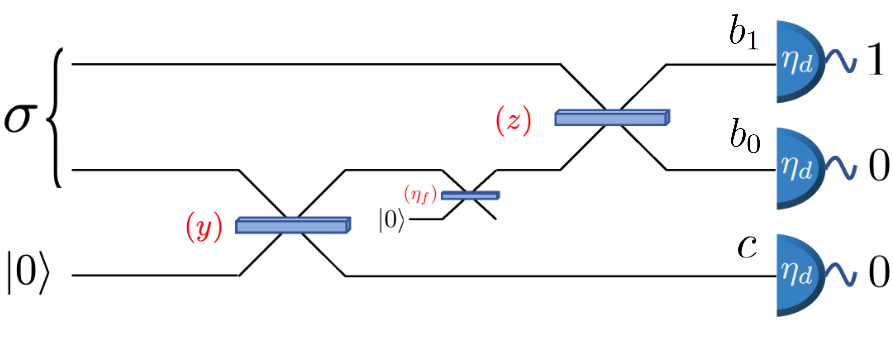}
		\caption{Alice aims to maximize the outcome $(1,0,0)$ by sending the state $\sigma$. The lossy delay line is represented by a mixing with the vacuum on a  beam splitter of transmission amplitude $\eta_f$. The quantum efficiency of the detectors is indicated in white.}
		\label{fig:lossy1}
	\end{center}
\end{figure}
\begin{figure}[h!]
	\begin{center}
		\includegraphics[width=4.3in]{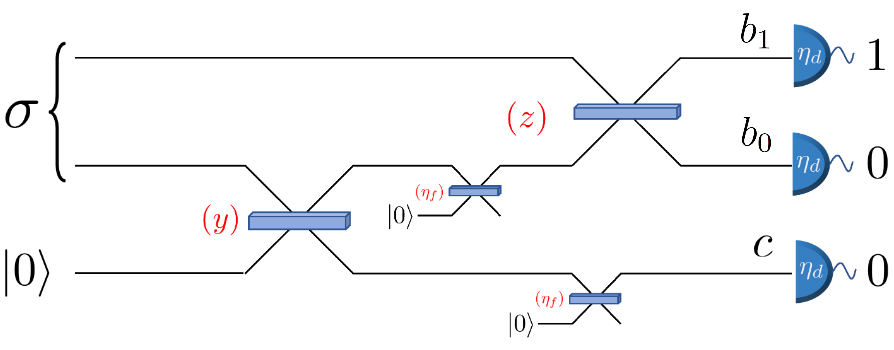}
		\caption{Adding losses on the third mode increases Alice's winning probability.}
		\label{fig:lossy2}
	\end{center}
\end{figure}
\begin{figure}[h!]
	\begin{center}
		\includegraphics[width=4.3in]{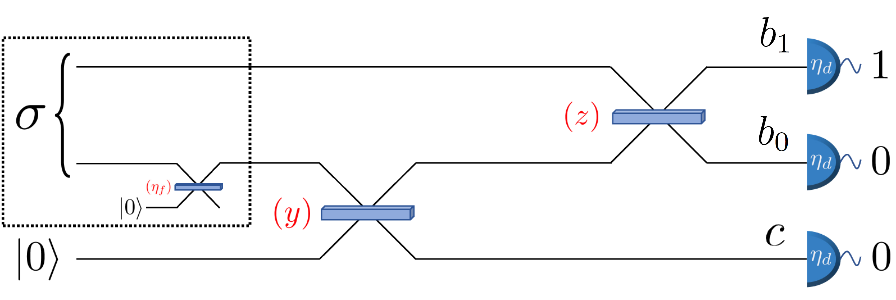}
		\caption{The losses $\eta_f$ are commuted back to Alice's state preparation. The losses on input mode $3$ can be omitted since the input state is the vacuum.}
		\label{fig:lossy3}
	\end{center}
\end{figure}
\begin{figure}[h!]
	\begin{center}
		\includegraphics[width=4.3in]{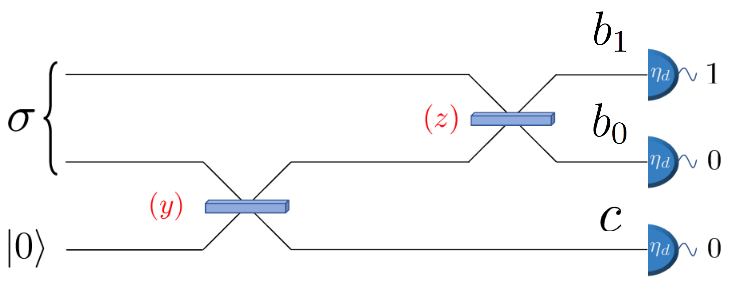}
		\caption{Alice aims to maximize the outcome $(1,0,0)$ by sending the state $\sigma$. The delay line efficiency $\eta_f$ is equal to $1$.}
		\label{fig:lossy4}
	\end{center}
\end{figure}

Since the input state on mode $3$ is the vacuum, the losses on this mode may then be removed (Fig.~\ref{fig:lossy3}).
In that case, the probability of winning is clearly lower than when the delay line is perfect (Fig.~\ref{fig:lossy4}), because Alice is now restricted to lossy state preparation instead of ideal state preparation.

This reduction shows that Alice’s maximum winning probability when Bob is using a lossy delay line is always lower than when Bob’s delay line is perfect, independently of the efficiency $\eta_d$ of his detectors.\\

Moreover, Alice's maximum cheating probability and optimal cheating strategy may be inferred from the case where Bob has a perfect delay line, as we show in what follows.
By convexity of the probabilities, Alice's best strategy is to send a pure state $\ket\psi=\sum_{k,l\geqslant0}{\psi_{kl}\ket{kl}}$. Let us denote by $W$ the interferometer depicted in Fig.~\ref{fig:lossy1}, including the detection losses. Let us consider the evolution of Alice's state and the vacuum on the third input mode through the interferometer $W$. The creation operator for the first mode evolves as
\be
\ba
\hat a_1^\dag&\rightarrow\sqrt z\hat a_1^\dag+\sqrt{1-z}\hat a_2^\dag\\
&\rightarrow\sqrt{z\eta_d}\hat a_1^\dag+\sqrt{(1-z)\eta_d}\hat a_2^\dag\\
&=W\hat a_1^\dag W^\dag,
\ea
\ee
while the creation operator for the second mode evolves as
\be
\ba
\hat a_2^\dag&\rightarrow\sqrt y\hat a_2^\dag+\sqrt{1-y}\hat a_3^\dag\\
&\rightarrow\sqrt{y\eta_f}\hat a_2^\dag+\sqrt{1-y}\hat a_3^\dag\\
&\rightarrow\sqrt{y(1-z)\eta_f}\hat a_1^\dag-\sqrt{yz\eta_f}\hat a_2^\dag+\sqrt{1-y}\hat a_3^\dag\\
&\rightarrow\sqrt{y(1-z)\eta_f\eta_d}\hat a_1^\dag-\sqrt{yz\eta_f\eta_d}\hat a_2^\dag+\sqrt{(1-y)\eta_d}\hat a_3^\dag\\
&=W\hat a_2^\dag W^\dag.
\ea
\ee
Hence, the output state (before the ideal threshold detection) is given by
\be
\ba
W\ket{\psi0}&=W\sum_{k,l\geqslant0}{\psi_{kl}\ket{kl0}}\\
&=W\left[\sum_{k,l\geqslant0}{\frac{\psi_{kl}}{\sqrt{k!l!}}(\hat a_1^\dag)^k(\hat a_2^\dag)^l}\right]\ket{000}\\
&=\left[\sum_{k,l\geqslant0}{\frac{\psi_{kl}}{\sqrt{k!l!}}(W\hat a_1^\dag W^\dag)^k(W\hat a_2^\dag W^\dag)^l}\right]\ket{000}\\
&=\left[\sum_{k,l\geqslant0}{\frac{\psi_{kl}}{\sqrt{k!l!}}(\sqrt{z\eta_d}\hat a_1^\dag+\sqrt{(1-z)\eta_d}\hat a_2^\dag)^k(\sqrt{y(1-z)\eta_f\eta_d}\hat a_1^\dag-\sqrt{yz\eta_f\eta_d}\hat a_2^\dag+\sqrt{(1-y)\eta_d}\hat a_3^\dag)^l}\right]\ket{000}.
\ea
\ee
Now Alice's maximum cheating probability is given by
\be
P_d^{(A)}=\mathrm{Tr}[W\ket{\psi0}\bra{\psi0}W^\dag(\mathbb1-\ket0\bra0)\ket{00}\bra{00}].
\ee
Hence, the state after a successful projection $(\mathbb1-\ket0\bra0)\ket{00}\bra{00}$, which has norm $P_d^{(A)}$, reads
\be
\left[\sum_{k+l>0}{\frac{\psi_{kl}}{\sqrt{k!l!}}(z\eta_d)^{k/2}[y(1-z)\eta_f\eta_d]^{l/2}(\hat a_1^\dag)^{k+l}}\right]\ket{000}.
\ee
When Bob has a perfect delay line ($\eta_f=1$) this state reads
\be
\left[\sum_{k+l>0}{\frac{\psi_{kl}}{\sqrt{k!l!}}(z\eta_d)^{k/2}[y(1-z)\eta_d]^{l/2}(\hat a_1^\dag)^{k+l}}\right]\ket{000},
\ee
and its norm is the winning probability of Alice in that case. Hence,
\be
P_d^{(A)}[\eta_f,\eta_d,y,z]=P_d^{(A)}[1,\eta_d,y\eta_f,z],
\label{etamPd}
\ee
i.e.\@ we can obtain Alice's cheating probability by solving the case with perfect delay line, and replacing the parameter $y$ by $y\eta_f$. In the following, we thus derive Alice's optimal strategy in that case.

\subsection{Perfect delay line}

Let $\sigma$ be the state sent by Alice, and $\eta_d$ the detector efficiency. She needs to maximize the probability of the overall outcome $(1,0,0)$ at the output of the interferometer depicted in Fig.~\ref{fig:eta1}, hence the overlap with the projector:

\begin{equation}
\Pi_{(1,0,0)}^{\eta_d}=\left[\mathbb{1}-\sum_m(1-\eta_d)^m\ket{m}\bra{m}\right]\otimes\left[\sum_{n,p}{(1-\eta_d)^{n+p}\ket{n}\bra{n}\otimes\ket{p}\bra{p}}\right].
\label{Pi_100eta}
\end{equation}

By convexity of the probabilities, we may assume without loss of generality that Alice sends a pure state $\sigma=\ket\psi\bra\psi$.
Moreover, the imperfect threshold detectors of quantum efficiency $\eta_d$ can be modelled by mixing the state to be measured with the vacuum on a beam splitter of transmission amplitude $\eta_d$ followed by an ideal threshold detection~\cite{ferraro2005gaussian}. In that case, this corresponds to losses $\eta_d$ on modes $1$, $2$, and $3$, followed by ideal threshold detections. By Lemma~\ref{lem:commut}, commuting the losses back through the interferometer leads to the equivalent picture depicted in Fig.~\ref{fig:eta2}, where the losses on input mode $3$ have been omitted, since the input state is the vacuum.

In that case, Alice's probability of winning is clearly lower than when the threshold detectors are perfect (Fig.~\ref{fig:Bhonest}), because she is restricted to lossy state preparation instead of ideal state preparation. Let $\ket{\tilde\psi}$ be the lossy state obtained by applying losses $\eta_d$ on both modes of Alice's prepared state $\ket\psi$. Alice's winning probability may then be written:
\begin{equation}
\ba
    P_d^{(A)}&=\mathrm{Tr}[U(\ket{\tilde\psi}\bra{\tilde\psi}\otimes\ket0\bra0)U^\dag(\mathbb{1}-\ket0\bra0)\otimes\ket{00}\bra{00}]\\
    &=\mathrm{Tr}[U(\ket{\tilde\psi}\bra{\tilde\psi}\otimes\ket0\bra0)U^\dag(\mathbb{1}\otimes\ket{00}\bra{00})]-\mathrm{Tr}[U(\ket{\tilde\psi}\bra{\tilde\psi}\otimes\ket0\bra0)U^\dag\ket{000}\bra{000}],
\ea
\end{equation}
where $U=(H^{(z)}\otimes \mathbb{1})(\mathbb{1}\otimes H^{(y)})$ is the unitary corresponding to the general interferometer of the lossless protocol. By Lemma~\ref{lem:reductionUV}, we have
\begin{equation}
\mathrm{Tr}[(\tau\otimes\ket0\bra0)U^\dag(\mathbb{1}\otimes\ket{00}\bra{00})U]=\mathrm{Tr}[(\tau\otimes\ket0\bra0)V^\dag(\ket{0}\bra{0}\otimes \mathbb{1}\otimes\ket{0}\bra{0})V],
\end{equation}
for any density matrix $\tau$, where $V=(\mathbb{1}\otimes H^{(b)})(H^{(a)}\otimes \mathbb{1})(\mathbb{1}\otimes R(\pi)\otimes \mathbb{1})$, with $a=\frac{y(1-z)}{y+z-yz}$ and $b=y+z-yz$, and $R(\pi)$ a phase shift of $\pi$ acting on mode $2$. Hence,
\begin{equation}
    P_d^{(A)}=\mathrm{Tr}[V(\ket{\tilde\psi}\bra{\tilde\psi}\otimes\ket0\bra0)V^\dag(\ket{0}\bra{0}\otimes \mathbb{1}\otimes\ket{0}\bra{0})]-\mathrm{Tr}[\ket{\tilde\psi}\bra{\tilde\psi}\ket{00}\bra{00}],
\end{equation}
where we used $U^\dag\ket{000}=\ket{000}$ for the second term. Setting $\ket{\tilde\psi_x}=(H^{(a)}\otimes \mathbb{1})(\mathbb{1}\otimes R(\pi))\ket{\tilde\psi}$ yields
\begin{equation}
    P_d^{(A)}=\underbrace{ \mathrm{Tr}[(\ket{\tilde\psi_x}\bra{\tilde\psi_x}\otimes\ket0\bra0)(\mathbb{1}\otimes H^{(b)})(\ket{0}\bra{0}\otimes \mathbb{1}\otimes\ket{0}\bra{0})(\mathbb{1}\otimes H^{(b)})]}_{\equiv P_1}-\underbrace{\mathrm{Tr}[\ket{\tilde\psi_x}\bra{\tilde\psi_x}\ket{00}\bra{00}]}_{\equiv P_2},
    \label{P1P2}
\end{equation}
where we used $\ket{00}=(\mathbb{1}\otimes R(\pi))H^{(a)}\ket{00}$ for the second term $P_2$.

Let us consider the first term $P_1$. Since $\ket{\tilde\psi}$ is the state obtained by applying losses $\eta_d$ on both modes of the state $\ket\psi$, we obtain the equivalent picture in Fig.~\ref{fig:eta3}, where we have added losses $\eta_d$ also on mode $3$, since the input state is the vacuum.

Let $\ket{\psi_x}=H^{(a)}(\mathbb{1}\otimes R(\pi))\ket\psi$. With Lemma~\ref{lem:commut}, commuting the losses $\eta_d$ to the output of the interferometer in Fig.~\ref{fig:eta3}, and combining the losses on mode $2$ and $3$ yields
\begin{equation}
    P_1=\mathrm{Tr}[\ket{\psi_x}\bra{\psi_x}\Pi_{(0)}^{\eta_d}\otimes\Pi_{(0)}^{\eta_d(1-b)}],
\end{equation}
where $\Pi_{(0)}^\eta$ is the POVM element corresponding to no click for a threshold detector of quantum efficiency $\eta$ (recall that this is the same as an ideal detector preceded by a mixing with the vacuum on a beam splitter of transmission amplitude $\eta$). The same reasoning for the second term $P_2$ gives
\begin{equation}
    P_2=\mathrm{Tr}[\ket{\psi_x}\bra{\psi_x}\Pi_{(0)}^{\eta_d}\otimes\Pi_{(0)}^{\eta_d}],
\end{equation}
and we finally obtain with Eq.~(\ref{P1P2}),
\begin{equation}
    P_d^{(A)}=\mathrm{Tr}[\ket{\psi_x}\bra{\psi_x}\Pi_{(0)}^{\eta_d}\otimes(\Pi_{(0)}^{\eta_d(1-b)}-\Pi_{(0)}^{\eta_d})].
\end{equation}
\begin{figure}[h!]
	\begin{center}
		\includegraphics[width=4in]{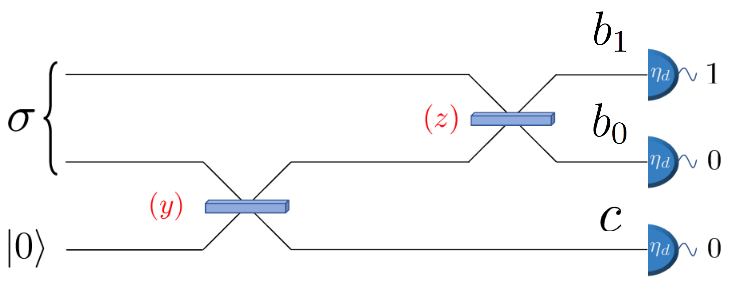}
		\caption{Alice aims to maximize the outcome $(1,0,0)$ by sending the state $\sigma$. The quantum efficiency of the detectors is indicated in white.}
		\label{fig:eta1}
	\end{center}
\end{figure}
\begin{figure}[h!]
	\begin{center}
		\includegraphics[width=4in]{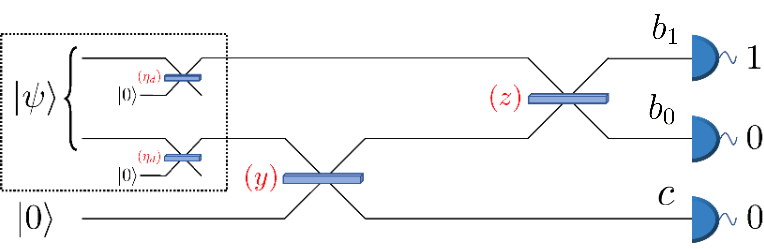}
		\caption{The quantum efficiency are modelled as losses $\eta_d$ on modes $1$, $2$, and $3$, which are then commuted through the interferometer, back to Alice's state preparation. The losses on input mode $3$ can be omitted since the input state is the vacuum.}
		\label{fig:eta2}
	\end{center}
\end{figure}
\begin{figure}[h!]
	\begin{center}
		\includegraphics[width=4in]{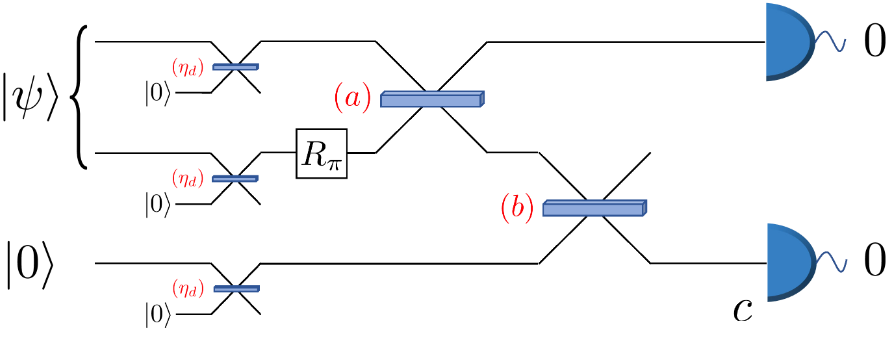}
		\caption{An equivalent picture for the first term $P_1$ of Eq.~(\ref{P1P2}). The term $P_1$ is the probability of the simultaneous outcomes $0$ for modes $1$ and $3$.}
		\label{fig:eta3}
	\end{center}
\end{figure}

Let us write $\ket{\psi_x}=\sum_{k,l\ge0}^{+\infty}{\psi_{kl}\ket{kl}}$. With the expression of the POVM in Eq.~(\ref{Pi_100eta}) the last equation reads
\be
\ba
P_d^{(A)}&=\sum_{k,l\ge0}{|\psi_{kl}|^2(1-\eta_d)^k[(1-\eta_d(1-b))^l-(1-\eta_d)^l]}\\
&\leqslant\max_{k,l\geqslant0}{(1-\eta_d)^k[(1-\eta_d(1-b))^l-(1-\eta_d)^l]}\sum_{k,l\geqslant0}{|\psi_{kl}|^2}\\
&=\max_{k,l\geqslant0}{(1-\eta_d)^k[(1-\eta_d(1-b))^l-(1-\eta_d)^l]}\\
&=\max_{l\geqslant1}{[(1-\eta_d(1-b))^l-(1-\eta_d)^l]}\\
&=\max_{l\geqslant1}{[(1-\eta_d(1-y)(1-z))^l-(1-\eta_d)^l]},
\ea
\ee
where we used $b=y+z-yz$. Let $l_0\in\mathbb N^*$ such that $\max_{l\geqslant1}{[(1-\eta_d(1-b))^l-(1-\eta_d)^l]}=(1-\eta_d(1-b))^{l_0}-(1-\eta_d)^{l_0}$. This last expression is an upperbound for $P_d^{(A)}$, which is attained for $\psi_{kl}=\delta_{k,0}\delta_{l,l_0}$, i.e.\@ $\ket{\psi_x}=\ket{0l_0}$.
Thus, the best strategy for Alice is to send the state
\be
\ba
\ket\psi&=(\mathbb{1}\otimes R(\pi))H^{(a)}\ket{\psi_x}\\
&=(\mathbb{1}\otimes R(\pi))H^{(a)}\ket{0l_0},
\ea
\ee
where $a=\frac{y(1-z)}{y+z-yz}$, and her winning probability is then
\be
P_d^{(A)}=(1-\eta_d(1-y)(1-z))^{l_0}-(1-\eta_d)^{l_0},
\ee
when Bob has a perfect delay line.
Recalling Eq.~(\ref{etamPd}), the best strategy for Alice when Bob has a lossy delay line of efficiency $\eta_f$ is to send the state
\be
\ba
\ket\psi&=(\mathbb{1}\otimes R(\pi))H^{(a)}\ket{\psi_x}\\
&=(\mathbb{1}\otimes R(\pi))H^{(a)}\ket{0l_1},
\ea
\ee
where $a=\frac{y(1-z)\eta_f}{y\eta_f+z-yz\eta_f}$, and $l_1\in\mathbb N^*$ maximizes $(1-\eta_d(1-y\eta_f)(1-z))^l-(1-\eta_d)^l$. Her winning probability is then
\be
\ba
P_d^{(A)}&=\max_{l>0}\left[\left(1-(1-y\eta_f)(1-z)\eta_d\right)^l-\left(1-\eta_d\right)^l\right]\\
&=(1-\eta_d(1-y\eta_f)(1-z))^{l_1}-(1-\eta_d)^{l_1}\\
&=\eta_d[1-(1-y\eta_f)(1-z)]\sum_{j=0}^{l_1-1}{(1-\eta_d)^j(1-\eta_d(1-y\eta_f)(1-z))^{l_1-j-1}}\\
&\leqslant\eta_d[1-(1-y\eta_f)(1-z)]\sum_{j=0}^{l_1-1}{(1-\eta_d)^j}\\
&=\eta_d[1-(1-y\eta_f)(1-z)]\frac{1-(1-\eta_d)^{l_1}}{1-(1-\eta_d)}\\
&=[1-(1-y\eta_f)(1-z)][1-(1-\eta_d)^{l_1}]\\
&\leqslant1-(1-y\eta_f)(1-z)\\
&\leqslant 1-(1-y)(1-z),
\label{eq:pda}
\ea
\ee
and this last expression is the winning probability when there are no losses.

\medskip

Let us derive the value of $l_1$. For this, we define:
\be
\ba
&r = 1-\eta_d(1-y\eta_f)(1-z)\\
&s = 1-\eta_d.\\
\ea
\ee
We then consider a $\lambda_1\in\mathbb{R}^{*+}$ which maximizes $(r^\lambda-s^\lambda)$ for $\lambda\in\mathbb{R}^{*+}$. We have that:
\be
\ba
&\frac{d}{d \lambda_1}(r^{\lambda_1}-s^{\lambda_1})=0 \Leftrightarrow \lambda_1 = \frac{\ln{\ln{s}}-\ln{\ln{r}}}{\ln{r}-\ln{s}},
\ea
\ee
for strictly non-zero $r$ and $s$ and where $ln$ denotes the complex logarithm function. This allows to deduce:
\begin{equation}
   l_1 = \left\{
    \begin{array}{ll}
     &\text{floor}(\lambda_1) \:\:\:\:\:\:\:\: \text{if} \:\:\:\:\:\:\:\: r^{\text{floor}(\lambda_1)}-s^{\text{floor}(\lambda_1)}\geqslant r^{\text{ceil}(\lambda_1)}-s^{\text{ceil}(\lambda_1)}\\
     &\text{ceil}(\lambda_1)\:\:\:\:\:\:\:\:\:\: \text{if} \:\:\:\:\:\:\:\: r^{\text{ceil}(\lambda_1)}-s^{\text{ceil}(\lambda_1)} \geqslant r^{\text{floor}(\lambda_1)}-s^{\text{floor}(\lambda_1)}.\\
    \end{array}
\right.
\end{equation}

%-------------------------------------------------------

\section{Solving the system from Eq.~(\ref{eq:systemz})}
\label{sec:system}

\subsection{Condition (i)}

The first condition enforces a fair protocol, i.e.\@ $P_{h}^{(A)}=P_{h}^{(B)}$. With Eqs.~(\ref{eq:phb}) and (\ref{eq:pha}), we aim to solve for $y$ as a function of $x$ and $z$:
\begin{equation}
    \begin{aligned}
   (i) \Leftrightarrow &\:\:\eta_t\eta_d^{(B)}\left(\sqrt{xz\eta_f^{(A)}}+\sqrt{(1-x)y(1-z)\eta_f^{(B)}}\right)^2 =  \eta_t\eta_d^{(B)}(1-x)(1-y)\\
  (i) \Leftrightarrow &\:\: (1-x)\left[(1-z)\eta_f^{(B)}+1\right]y+2\sqrt{x(1-x)z(1-z)\eta_f^{(A)}\eta_f^{(B)}}\sqrt{y}+xz\eta_f^{(A)}-(1-x)=0.
    \end{aligned}
    \label{eqeq}
\end{equation}
We make the substitution $Y = \sqrt{y}$ in order to transform Eq. (\ref{eqeq}) into a second-order polynomial equation. We then take only the positive solution (since $y$ must be positive) which reads:
\begin{equation}
    Y = \frac{\sqrt{xz(1-z)\eta_f^{(A)}\eta_f^{(B)}-\left[(1-z)\eta_f^{(B)}+1\right]\left[xz\eta_f^{(A)}-(1-x)\right]}-\sqrt{xz(1-z)\eta_f^{(A)}\eta_f^{(B)}}}{\sqrt{1-x}\left[(1-z)\eta_f^{(B)}+1\right]}.
    \label{eq:bigY}
\end{equation}
We may finally write:
\begin{equation}
    (i) \Leftrightarrow \:\: y=f\left(x,z,\eta_f^{(i)},\eta_d,\eta_t\right),
\end{equation}
where $f\left(x,z,\eta_f^{(i)},\eta_d,\eta_t\right)=\frac{\left(\sqrt{(1-x)\left[(1-z)\eta_f^{(B)}+1\right]-xz\eta_f^{(A)}}-\sqrt{xz(1-z)\eta_f^{(A)}\eta_f^{(B)}}\right)^2}{(1-x)\left[(1-z)\eta_f^{(B)}+1\right]^2}$.

Note that $y$ should be a real number, and hence we require that the expression under the first square root of $f\left(x,z,\eta_f^{(i)},\eta_d,\eta_t\right)$ is positive, i.e.:
\begin{equation}
    z\leqslant \frac{(1-x)(1+\eta_f^{(B)})}{x\eta_f^{(A)}+(1-x)\eta_f^{(B)}}.
\end{equation}
Furthermore, note that, for $\eta_f^{(A)}=\eta_f^{(B)}=\eta_f$, $y$ should be an increasing function of $\eta_f$, and therefore a decreasing function of $d$ when $\eta_f=10^{-\frac{0.2}{10}2d}$. Mathematically speaking, this is to prevent $y'(d)\rightarrow \infty$ and $y(d)>1$. Physically speaking, this condition ensures that, as the probability of transmitting the photon (and of preserving it for verification) gets smaller, Bob should encourage a detection on the third mode, which evens out the honest probabilities of winning.

\subsection{Condition (ii)}

The second condition enforces a balanced protocol, i.e.\@ $P_d^{(A)}=P_d^{(B)}$. With Eqs.~(\ref{eq:pdb}) and (\ref{eq:pda}), this translates into the following expression for $x$:
\begin{equation}
    (ii) \Leftrightarrow  x=g\left(y,z,\eta_f^{(i)},\eta_d^{(i)}\right),
\end{equation}
where
\begin{equation}
g\left(y,z,\eta_f^{(i)},\eta_d^{(i)}\right)=\frac1{\eta_f^{(A)}\eta_d^{(A)}}\left[1-\max_{l\geqslant1}{[(1-\eta_d^{(B)}(1-y\eta_f^{(B)})(1-z))^l-(1-\eta_d^{(B)})^l]}\right].
\end{equation}

\subsection{Condition (iii)}

We recall the general coin flipping formalism from \cite{HW:TCC11}, in which any classical or quantum coin flipping protocol may be expressed as:

\begin{equation}
    CF\left(p_{00},p_{11},p_{*0},p_{*1},p_{0*},p_{1*}\right),
\end{equation}
where $p_{ii}$ is the probability that two honest players output value $i\in\{0,1\}$, $p_{*i}$ is the probability that Dishonest Alice forces Honest Bob to declare outcome $i$, and $p_{i*}$ is the probability that Dishonest Bob forces Honest Alice to declare outcome $i$. In this formalism, a perfect SCF protocol can then be expressed as  $CF\left(\frac{1}{2},\frac{1}{2},\frac{1}{2},\frac{1}{2},\frac{1}{2},\frac{1}{2}\right)$, while a perfect WCF may be expressed as  $CF\left(\frac{1}{2},\frac{1}{2},\frac{1}{2},1,1,\frac{1}{2}\right)$. We may now express our quantum WCF protocol in the lossless setting as:
\begin{equation}
    CF\left(\frac{1}{2},\frac{1}{2},\left[\frac{1}{2(1-x)}\right],1,1,[1-x]\right).
\end{equation}
In the lossy setting, note that the probabilities that Alice and Bob each choose to lose (i.e. $p_{*1}$ and $p_{0*}$, respectively), both remain $1$. When Dishonest Bob chooses to lose, he may always declare outcome $0$ regardless of what he detects, which yields $p_{0*}=1$. When Dishonest Alice chooses to lose, she may send a state $\ket n$ to Bob, and so:
\be
\ba
p_{*1}&=\mathrm{Tr}\left[H^{(y)}\ket{n0}\bra{n0}H^{(y)}I\otimes(I-\Pi_0)\right]\\
&=1-\mathrm{Tr}\left[H^{(y)}\ket{n0}\bra{n0}H^{(y)}(I\otimes\Pi_0)\right],
\ea
\ee
where $\Pi_0=\sum_{l\ge0}(1-\eta)^l\ket l\bra l$ and $H^{(y)}=\begin{pmatrix}
\sqrt y & \sqrt {1-y}\\\sqrt{1-y} & -\sqrt y
\end{pmatrix}$. 

\medskip

\noindent Now,
\be
\ba
H^{(y)}\ket{n0}&=H^{(y)}\frac{(\hat a_1^\dag)^n}{\sqrt{n!}}\ket{00}\\
&=\frac1{\sqrt{n!}}(\sqrt y\hat a_1^\dag+\sqrt{1-y}\hat a_2^\dag)^n\ket{00}\\
&=\frac1{\sqrt{n!}}\sum_{k=0}^n\binom nky^{\frac k2}(1-y)^{\frac{n-k}2}\hat a_1^{\dag k}\hat a_2^{\dag(n-k)}\ket{00}\\
&=\sum_{k=0}^n\sqrt{\binom nky^k(1-y)^{n-k}}\ket{k\text{ }(n-k)}.
\ea
\ee
We thus obtain, by linearity of the trace:
\be
\ba
p_{*1}&=1-\sum_{l,l'\ge0}(1-\eta)^l\sum_{k,k'=0}^n\sqrt{\binom nky^k(1-y)^{n-k}}\sqrt{\binom n{k'}y^{k'}(1-y)^{n-k'}}\mathrm{Tr}\left[\ket{k\text{ }(n-k)}\bra{k'\text{ }(n-k')}\ket{l'l}\bra{l'l}\right]\\
&=1-\sum_{k=0}^n(1-\eta)^{n-k}\binom nky^k(1-y)^{n-k}\\
&=1-\left[y+(1-\eta)(1-y)\right]^n,
\ea
\ee
which goes to $1$ when $n$ goes to infinity, for $y<1$. Hence, in the lossy setting, the protocol becomes a:
\begin{equation}
    CF\left(P_{h}^{(A)},P_{h}^{(B)},P_d^{(A)},1,1,P_d^{(B)}\right),
\end{equation}
where $P_d^{(A)}=\max_{l>0}\left(1-(1-y\eta_f^{(A)})(1-z)\eta_d^{(B)}\right)^l-\left(1-\eta_d^{(B)}\right)^l$ and $P_d^{(B)} = 1-x\eta_f^{(A)}\eta_d^{(A)}$.

Using Theorem $1$ from \cite{HW:TCC11}, there exists a classical protocol that implements an information-theoretically secure coin flip with our parameters if and only if the following conditions hold:

\be
\begin{cases}
P_h^{(A)}\le P_d^{(A)}\\
P_h^{(B)}\le P_d^{(B)}\\
P_{ab}=1-P_h^{(A)}-P_h^{(B)}\ge(1-P_d^{(A)})(1-P_d^{(B)}).
\label{system_qc}
\end{cases}
\ee
%
%\begin{equation}
%\left\{
 %   \begin{array}{ll}
  %     P_{h}^{(A)}\leqslant P_{d}^{(A)} \\
   %    P_{h}^{(B)}\leqslant (1-y)\eta_d^{(B)}P_{d}^{(B)} \\
    %   P_{h}^{(A)}+P_{h}^{(B)}\leqslant P_{d}^{(A)}+(1-y)\eta_d^{(B)}P_{d}^{(B)}-P_{d}^{(B)}\max\left(0, P_{d}^{(A)}+(1-y)P_{d}^{(B)}-1\right)
    %\end{array}
%\right.
%\label{system_qc}
%\end{equation}
%
Our quantum protocol therefore presents an advantage over classical protocols if at least one of these conditions \textit{cannot} be satisfied. Since we are interested in fair and balanced protocols, setting $P_{h}=P_{h}^{(A)}=P_{h}^{(B)}$ and $P_{d}=P_{d}^{(A)}=P_{d}^{(B)}$ allows to rewrite (\ref{system_qc}) as: 
\be
\begin{cases}
P_h\le P_d\\
P_{ab}=1-2P_h\ge(1-P_d)^2 \Leftrightarrow P_h\le\frac12[1-(1-P_d)^2].
\end{cases}
\ee
Let us finally remark that for all $x$ we have $\frac12[1-(1-x)^2]=x-\frac{x^2}2\le x$, so the first inequality above is implied by the second. The system is thus equivalent to the second inequality:
\be
P_{ab}=1-2P_h\ge(1-P_d)^2,
\ee
provided that $P_h^{(A)}=P_h^{(B)}=P_h$ and $P_d^{(A)}=P_d^{(B)}=P_d$.

In order to get a clearer insight into the meaning of quantum advantage, we express this condition in terms of cheating probability: our protocol displays quantum advantage if and only if the lowest classical cheating probability

\be
P_d^{C}=1-\sqrt{1-2P_h}=1-\sqrt{P_{ab}}
\ee
exceeds our quantum cheating probability $P_d^{Q}$. 

%\begin{equation}
%\left\{
  %  \begin{array}{ll}
    %   P_{h}^{(B)}\leqslant (1-y)\eta_d^{(B)}P_{d}^{(B)} \\
    %   2P_{h}^{(B)}\leqslant P_{d}^{(B)}+(1-y)\eta_d^{(B)}P_{d}^{(B)}-P_{d}^{(B)}\max\left(0, P_{d}^{(B)}+(1-y)\eta_{d}^{(B)}-1\right)
   % \end{array}
%\right.
%\end{equation}
%
%which in turn yields:

%\begin{equation}
%\left\{
   % \begin{array}{ll}
    %   P_{d}^{(B)}\geqslant \frac{P_{h}^{(B)}}{(1-y)\eta_d^{(B)}}  \\
    %   P_{d}^{(B)}\geqslant \frac{2P_{h}^{(B)}}{1+(1-y)\eta_d^{(B)}}\:\:\:\:\:\: &if\:\:\:\:\:\: P_{d}^{(B)}+(1-y)\eta_{d}^{(B)}\leqslant1  \\
    %   P_{d}^{(B)}\geqslant 1-\sqrt{1-2P_{h}^{(B)}}\:\:\:\:\:\: \:\:\:\:\:\:\:\:\:\:\:\: &if\:\:\:\:\:\: P_{d}^{(B)}+(1-y)\eta_{d}^{(B)}>1 \\
    %\end{array}
%\right.
%\end{equation}
%
%For a specific set of parameters, the best classical protocol saturates the largest bound of these two, since it must yield the lowest cheating probability. We have the condition (Q>C) whenever one of the following expressions for the lowest classical cheating probability exceeds our quantum cheating probability. 

\section{Practical quantum advantage for various detection efficiencies}
\label{sec:eta}

In this section, we plot the numerical solutions to the system  from Eq.~(\ref{eq:systemz}) in order to display quantum advantage as a function of distance for various detection efficiencies.  Numerical values for the lowest classical and quantum cheating probabilities, $P_d^{C}$ and $P_d^{Q}$, are plotted as a function of distance $d$ in blue and red, respectively. Our quantum protocol performs strictly better than any classical protocol when  $P_d^{Q}<P_d^{C}$. We set $\eta_f=\eta_s\eta_t^2$, where $\eta_s$ is the fiber delay transmission corresponding to $500$ns of optical switching time, and $\eta_t^2=\left(10^{-\frac{0.2}{10}d}\right)^2$ is the fiber delay transmission associated with travelling distance $d$ twice (once for quantum, once for classical) in single-mode fibers with attenuation $0.2$ dB/km.

\begin{figure}
	\begin{center}
		\includegraphics[width=0.92\columnwidth]{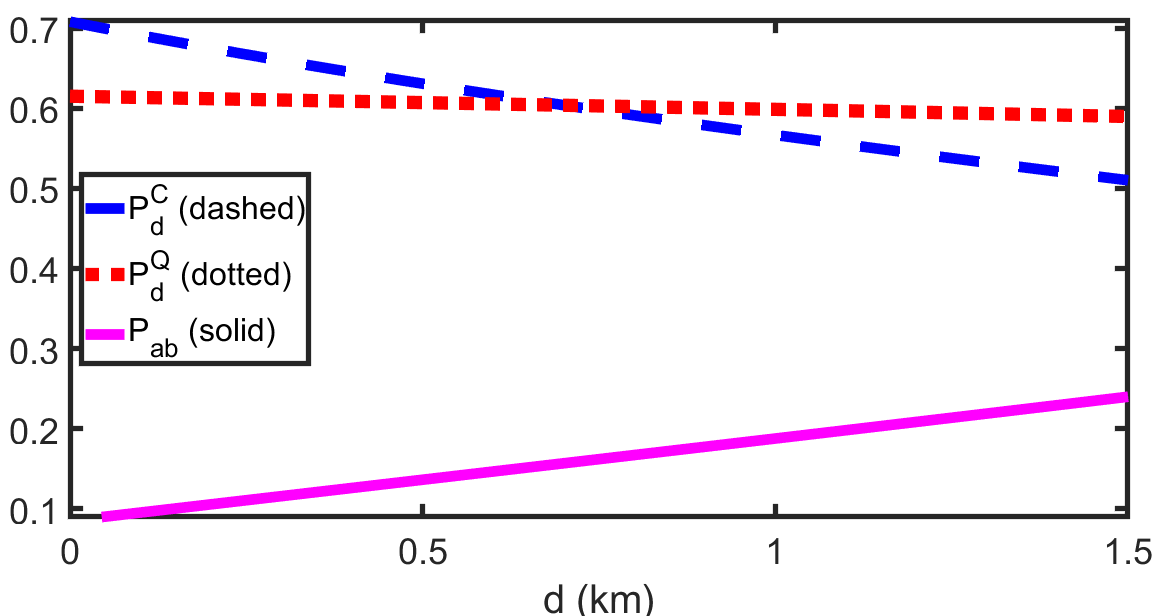}
		\caption{Parameters $\eta_d=0.95$ and $z=0.57$. Note that honest abort probability $P_{ab}$ is plotted in magenta. }
		\label{fig:performance1}
	\end{center}
\end{figure}

\begin{figure}
	\begin{center}
		\includegraphics[width=0.92\columnwidth]{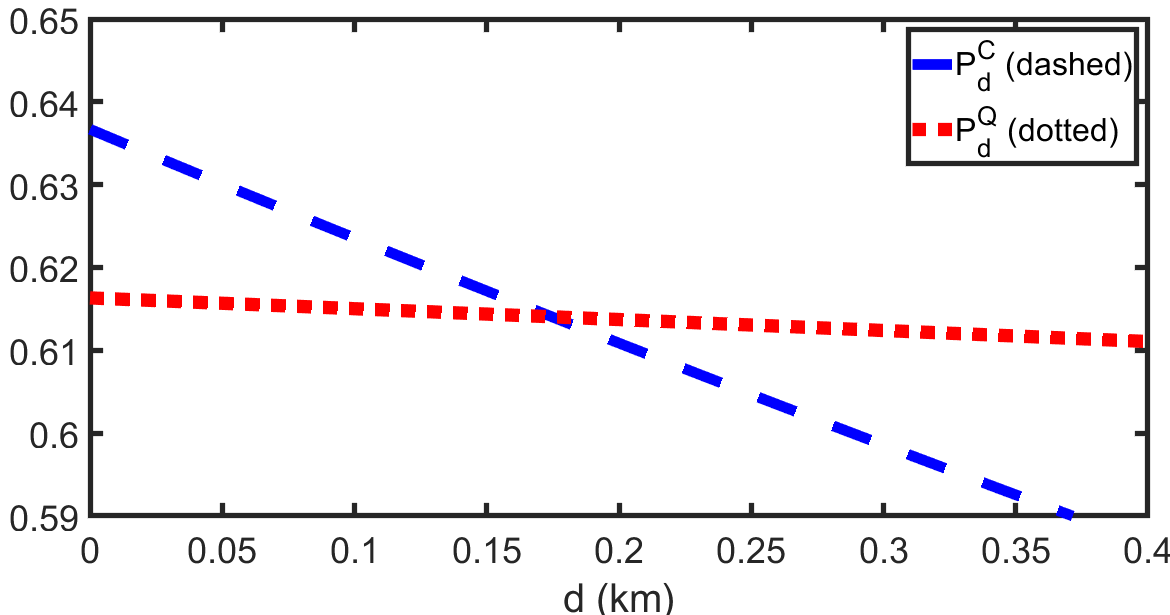}
		\caption{Parameters $\eta_d=0.90$ and $z=0.63$. Note that honest abort probability has been omitted in order to zoom in, but it lies around $0.15$ for these distances.}
		\label{fig:performance2}
	\end{center}
\end{figure}

%\end{widetext}
\end{appendix}

\end{document}